\newif\ifanon%
\anonfalse%
\newcommand{\anontext}[2]{\ifanon#1\else#2\fi}

\newif\ifarxivsubmit%
\arxivsubmittrue%
\newcommand{\arxivsubmittext}[2]{\ifarxivsubmit#1\else#2\fi}

\documentclass[runningheads]{llncs}
\usepackage[T1]{fontenc}

\usepackage[bookmarksnumbered,unicode]{hyperref}
\hypersetup{
  colorlinks=true, 
  urlcolor=blue, 
  linkcolor=blue, 
  citecolor=red 
}

\usepackage{breakcites}
\usepackage{amsmath}
\usepackage{amssymb}

\usepackage{amsthm}

\usepackage[inline]{enumitem}
\setlist[enumerate]*{label=(\arabic*)}
\setlist[itemize]{label=\textbullet}

\usepackage{xspace}
\usepackage{stmaryrd}
\usepackage{mathtools}

\usepackage{cite}
\usepackage[capitalise]{cleveref}

\makeatletter
\newcommand*{\crefns}[1]{{\@cref@sortfalse\cref{#1}}}
\makeatother

\crefname{definition}{Def.}{Defs.}

\usepackage{tasks}
\hypersetup{hypertexnames=false}
\newcommand{\labelparenths}[1]{(#1)}
\newcommand{\labelparenthstextup}[1]{\textup{\labelparenths{\textbf{#1}}}}
\newskip\listisep
\listisep\smallskipamount
\settasks{before-skip=\listisep,after-skip=\listisep,label-width=2em, label-align=right, label-format=\labelparenths, label=\arabic*}

\Crefname{assumption}{Assumption}{Assumptions}
\creflabelformat{assumption}{\textup{(#2#1#3)}}
\NewTasksEnvironment[
  counter=assumption,
  label=\arabic*,
  label-format=\labelparenthstextup,
  ref=\arabic*
]{assumes}[\item]

\Crefname{expectation}{Expectation}{Expectations}
\creflabelformat{expectation}{\textup{(#2#1#3)}}
\NewTasksEnvironment[
  counter=expectation,
  label=\arabic*,
  ref=\arabic*
]{expect}[\item]

\Crefname{requirement}{Requirement}{Requirements}
\creflabelformat{requirement}{\textup{(#2#1#3)}}
\NewTasksEnvironment[
  counter=requirement,
  label=\alph*,
  ref=\alph*
]{require}[\item]

\Crefname{condition}{Condition}{Conditions}
\creflabelformat{condition}{\textup{(#2#1#3)}}
\NewTasksEnvironment[
  counter=condition,
  label=\alph*,
  ref=\alph*
]{condition}[\item]

\Crefname{monocondition}{Monotonicity Condition}{Monotonicity Conditions}
\creflabelformat{monocondition}{\textup{(#2#1#3)}}
\NewTasksEnvironment[
  counter=monocondition,
  label=\roman*,
  ref=\roman*
]{monocondition}[\item]

\Crefname{step}{Step}{Steps}
\creflabelformat{step}{\textup{(#2#1#3)}}
\NewTasksEnvironment[
  counter=step,
  label=\arabic*,
  ref=\arabic*
]{step}[\item]

\usepackage[labelformat=simple]{subcaption}

\usepackage{environ}
\NewEnviron{code}{%
  \begingroup%
  \setlength{\jot}{0pt}%
  \begin{displaymath}%
    \BODY
  \end{displaymath}%
  \endgroup%
}
\NewEnviron{codealign}{%
  \begingroup%
  \setlength{\jot}{0pt}%
  \begin{flalign*}%
    \BODY
  \end{flalign*}%
  \endgroup%
}

\usepackage{adjustbox}

\usepackage{tikz}
\tikzset{>=stealth}
\usetikzlibrary{calc}
\usetikzlibrary{math}
\usetikzlibrary{positioning}
\usetikzlibrary{quotes}
\usetikzlibrary{shapes}
\usetikzlibrary{arrows.meta}
\usetikzlibrary{decorations.pathmorphing}
\usetikzlibrary{decorations.markings}

\newdimen\XCoord
\newdimen\YCoord

\tikzmath{\defaultbend=10;}
\tikzmath{\defaulttypeheight=4.5cm;}
\tikzmath{\defaulttypewidth=2.5cm;}
\tikzmath{\defaulttypedistance=3.4;}
\tikzmath{\defaulttypedistancey=5.2;}
\tikzmath{\smallbigtypedistance=0.8*\defaulttypedistance;}
\tikzmath{\labeldistance=0;}
\tikzmath{\typeequivlabeldistance=0.1;}
\tikzmath{\labelarcradius=0.15cm;}
\tikzmath{\equivdistance=0.9;}
\tikzset{%
  none/.style={},
  dot/.style={circle, fill, minimum size=#1, inner sep=0pt, outer sep=0pt},
  dot/.default=3pt,
  leftstyle/.style={red},
  rightstyle/.style={blue},
  leftstyle2/.style={orange},
  rightstyle2/.style={violet},
  transfun/.style={dash pattern=on 5pt off 3pt,semithick},
  element/.style={dot={#1},draw=black,fill=black},
  relstyle/.style={draw,circle,thick,line cap=round, dash pattern=on 0pt off 2pt},
  loopstyle/.style={looseness=60},
  relstyle2base/.style={semithick,dash dot},
  relstyle2/.style={relstyle,relstyle2base},
  leftlabelstyle2/.style={leftstyle2,relstyle2base},
  rightlabelstyle2/.style={rightstyle2,relstyle2base},
  equivclass/.style={draw,circle,semithick,densely dotted,decoration={markings, mark=at position 0.28 with {\arrow{>}}},postaction={decorate}},
  type/.style={draw,ellipse,minimum width=\defaulttypewidth,minimum height=\defaulttypeheight},
  typesmall/.style={type,minimum width=0.6*\defaulttypewidth},
  typebig/.style={type,minimum width=1.4*\defaulttypewidth},
  edgestyle/.style={transfun,bend left=#1,midway,above},
  edgestyle/.default=\defaultbend,
  rightedge/.default=\defaultbend,
  proofedge/.style={decorate,decoration={zigzag}},
  asmedge/.style={edgestyle,solid,black!10!gray},
  typeequivclass/.pic={%
    \node[element] (element) {};
    \node[#1,equivclass] (equivclass) at (element.center) {};
  },
  typeequivclassleft/.pic={%
    \pic {typeequivclass={leftstyle}};
  },
  typeequivclassright/.pic={%
    \pic {typeequivclass={rightstyle}};
  },
  pics/equivclass/.style n args={3}{code ={%
    \draw[equivclass] let \p1=(#1), \p2=(#2), \n1={atan2(\y2-\y1,\x2-\x1)}, \n2={veclen(\y2-\y1,\x2-\x1)}
      in ($ (\p1)!0.5!(\p2) $) ellipse [x radius=\n2/2, y radius=#3, rotate=-\n1];
  }},
  pics/labelrel/.style n args={2}{code ={%
    \node[equivclass,#1,opacity=0] (equivclass) {};
    \draw[relstyle,#1,->] ([xshift=0.15cm]equivclass) arc (0:300:0.15cm);
    \node[#1,right=\labeldistance of equivclass] (correspondsto) {$\correspondsto$ #2};
  }},
  pics/typeshift/.style n args={3}{code ={%
    \node[type] (type) {};
    \pic (label) [#1,above=\labeldistance of type,xshift={#3}] {labelrel={#1}{#2}};
  }},
  pics/type/.style n args={2}{code ={%
    \pic {typeshift={#1}{#2}{-0.6cm}};
  }},
  lefttype/.pic={%
    \pic {type={leftstyle}{#1}};
  },
  righttype/.pic={%
    \pic {type={rightstyle}{#1}};
  },
  pics/rel/.style n args={2}{code ={%
    \node[draw,ellipse,minimum width=#1,minimum height=#1] (rel) {};
    \node[above=\labeldistance of rel] {#2};
  }},
  pics/leftrel/.style n args={2}{code ={%
    \pic[leftstyle]{rel={#1}{#2}};
  }},
  pics/rightrel/.style n args={2}{code ={%
    \pic[rightstyle]{rel={#1}{#2}};
  }},
  pics/leftedgestyled/.style n args={4}{code ={%
    \draw[->] (#1) edge[leftstyle,edgestyle,"#3",#4] (#2);
  }},
  pics/leftedge/.style n args={3}{code ={%
    \pic {leftedgestyled={#1}{#2}{#3}{}};
  }},
  pics/rightedgestyled/.style n args={4}{code ={%
    \draw[->] (#1) edge[rightstyle,edgestyle,"#3",#4] (#2);
  }},
  pics/rightedge/.style n args={3}{code ={%
    \pic {rightedgestyled={#1}{#2}{#3}{}};
  }},
  pics/typeequivedgepair/.style n args={4}{code ={%
    \pic {leftedge={#1}{#2}{#3}};
    \pic {rightedge={#2}{#1}{#4}};
  }}
}
\newcommand{\fitellipsis}[2] 
{\draw let \p1=(#1), \p2=(#2), \n1={atan2(\y2-\y1,\x2-\x1)}, \n2={veclen(\y2-\y1,\x2-\x1)}
    in ($ (\p1)!0.5!(\p2) $) ellipse [x radius=\n2/2, y radius=1cm, rotate=-\n1];
}

\newcommand{\citet}{\cite}
\newcommand{\Description}[2][]{}

\newcommand{\appendixref}[1]{\arxivsubmittext{\cref{#1}}{\cite[\begin{NoHyper}\hypersetup{hidelinks}\cref{#1}\end{NoHyper}]{transportarxiv}}}
\newcommand{\correspondsto}{\mathrel{\hat{=}}}

\newcommand{\textiff}{if and only if\xspace}

\newcommand{\transport}{\textsc{Transport}\xspace}
\newcommand{\app}{\,}
\newcommand{\inv}[1]{#1^{-1}}
\DeclarePairedDelimiter\parenths{(}{)}

\newcommand{\define}{\coloneqq}

\newcommand{\constfont}[1]{\mathsf{#1}}
\newcommand{\mathcodefont}[1]{\mathtt{#1}}
\newcommand{\isacmdfont}[1]{\mathcodefont{\textbf{#1}}}

\newcommand{\transportterm}{\isacmdfont{trp}}
\newcommand{\andcmd}{\isacmdfont{and}}
\newcommand{\wherecmd}{\isacmdfont{where}}
\newcommand{\bycmd}{\isacmdfont{by}}
\newcommand{\transporttermprover}{\isacmdfont{trprover}}

\newcommand{\lemmacmd}{\isacmdfont{lemma}}
\DeclarePairedDelimiterX\isatag[1]{\mathcodefont{[}}{\mathcodefont{]}}{#1}
\newcommand{\perintrotag}{\isatag{\mathcodefont{per\_intro}}}

\newcommand{\transpparameq}[1]{\mathcodefont{#1 ={}}}
\newcommand{\transfer}{\isacmdfont{transfer}}

\newcommand{\typefont}[1]{\mathsf{#1}}
\newcommand{\purefun}{\mathrel\typefont{\Rightarrow}}
\newcommand{\bool}{\typefont{bool}}
\newcommand{\alphaty}{\typefont{\alpha}}
\newcommand{\betaty}{\typefont{\beta}}
\newcommand{\gammaty}{\typefont{\gamma}}

\newcommand{\iarrayty}[1]{#1\app\typefont{iarray}}
\newcommand{\listty}[1]{#1\app\typefont{list}}
\newcommand{\mysetty}{\typefont{set}}
\newcommand{\settyconst}{\typefont{set}}
\newcommand{\setty}[1]{#1\app\settyconst}
\newcommand{\fsettyconst}{\typefont{fset}}
\newcommand{\fsetty}[1]{#1\app\fsettyconst}
\newcommand{\natty}{\typefont{\mathbb{N}}}
\newcommand{\binty}{\typefont{Bin}}
\newcommand{\intty}{\typefont{\mathbb{Z}}}
\newcommand{\inttycopy}{{\typefont{\mathbb{Z}'}}}
\newcommand{\relty}[2]{#1\purefun#2\purefun\bool}

\newcommand{\True}{\constfont{True}}
\newcommand{\holimplies}{\longrightarrow}

\newcommand{\holiff}{\longleftrightarrow}
\newcommand{\holhasty}{:}
\newcommand{\indom}{\constfont{in\_dom}}
\newcommand{\incodom}{\constfont{in\_codom}}
\newcommand{\infield}{\constfont{in\_field}}

\newcommand{\comp}{\circ}
\newcommand{\relcomp}{\mathbin{{\circ}}}
\newcommand{\finite}{\constfont{finite}}
\newcommand{\length}{\constfont{length}}
\newcommand{\tonat}{\constfont{to\_nat}}
\newcommand{\toint}{\constfont{to\_int}}

\newcommand{\intspos}{\constfont{Zpos}}
\newcommand{\listrel}{\constfont{ListRel}}

\newcommand{\iarraylength}{\constfont{iarr\_length}}
\newcommand{\iarrayrel}{\constfont{IArrRel}}
\newcommand{\iarrayindex}{\constfont{iarr\_ind}}

\newcommand{\toiarray}{\constfont{to\_iarr}}
\newcommand{\tolistset}{\constfont{to\_list}}
\newcommand{\tolistfset}{\constfont{to\_list^{fin}}}

\newcommand{\tolistiarray}{\constfont{to\_list}}
\newcommand{\toset}{\constfont{to\_set}}
\newcommand{\tofset}{\constfont{to\_fset}}
\newcommand{\listsetrel}{\constfont{LS}}
\newcommand{\listsetrelL}{\constfont{LS_L}}
\newcommand{\listfsetrel}{\constfont{LFS}}
\newcommand{\listfsetrelL}{\constfont{LFS_L}}

\newcommand{\intnatrel}{\constfont{ZN}}

\renewcommand{\max}{\constfont{max}}
\newcommand{\maxlist}{\constfont{max\_list}}
\newcommand{\maxset}{\constfont{max\_set}}
\newcommand{\maxfset}{\constfont{max\_fset}}
\newcommand{\foldr}{\constfont{foldr}}
\newcommand{\map}{\constfont{map}}

\newcommand{\filter}{\constfont{filter}}
\newcommand{\mapfset}{\constfont{map\_fset}}
\newcommand{\filterfset}{\constfont{filter\_fset}}
\newcommand{\intersect}{\constfont{intersect}}
\newcommand{\intersectfset}{\constfont{intersect\_fset}}

\newcommand{\relif}{\constfont{rel\_if}}

\newcommand{\inflistidx}{\constfont{(!!)}}

\newcommand{\lerell}[2]{\mathrel{{}_{#1}{\le_{#2}}}}
\newcommand{\lerel}[1]{\lerell{}{#1}}
\DeclarePairedDelimiterX\inflerell[2]{(}{)}{{}_{#1}{\le_{#2}}}
\DeclarePairedDelimiterX\inflerel[1]{(}{)}{{\le_{#1}}}

\DeclarePairedDelimiterX\infgerell[2]{(}{)}{{}_{#1}{\ge_{#2}}}
\DeclarePairedDelimiterX\infgerel[1]{(}{)}{{\ge_{#1}}}

\newcommand{\reflon}{\constfont{reflexive\_on}}
\newcommand{\refl}{\constfont{reflexive}}
\newcommand{\transon}{\constfont{transitive\_on}}
\newcommand{\trans}{\constfont{transitive}}
\newcommand{\symmon}{\constfont{symmetric\_on}}

\newcommand{\preorderon}{\constfont{preorder\_on}}

\newcommand{\partequivon}{\constfont{partial\_equivalence\_rel\_on}}
\newcommand{\partequiv}{\constfont{partial\_equivalence\_rel}}

\newcommand{\inflaton}{\constfont{inflationary\_on}}
\newcommand{\deflaton}{\constfont{deflationary\_on}}
\newcommand{\relequivon}{\constfont{rel\_equivalence\_on}}
\newcommand{\unitconst}{\constfont{unit}}
\newcommand{\unit}{\constfont{\eta}}
\newcommand{\counitconst}{\constfont{counit}}
\newcommand{\counit}{\constfont{\epsilon}}
\newcommand{\orderequivsym}{\equiv_{\constfont{o}}}
\DeclarePairedDelimiterXPP\orderequiv[4]{}{(}{)}{\app#3\app#4}{#1 \orderequivsym#2}

\newcommand{\depfunrelcolons}{::}
\newcommand{\relarrow}{\Rrightarrow}
\DeclarePairedDelimiterXPP\depfunrel[4]{}{[}{]}{\relarrow#4}{#1\app#2\depfunrelcolons#3}
\DeclarePairedDelimiterXPP\depfunrelrest[5]{}{[}{]}{\relarrow#5}{#1\app#2\depfunrelcolons#3\mid#4}

\newcommand{\funrel}[2]{#1\relarrow#2}
\newcommand{\monoarrow}{\Rrightarrow_{\constfont{m}}}

\DeclarePairedDelimiterXPP\depmono[4]{}{[}{]}{\monoarrow#4}{#1\app#2\depfunrelcolons#3}
\DeclarePairedDelimiterXPP\depmonorest[5]{}{[}{]}{\monoarrow#5}{#1\app#2\depfunrelcolons#3\mid#4}
\newcommand{\mono}[2]{#1\monoarrow#2}

\DeclarePairedDelimiterXPP\depfunrelpred[3]{}{[}{]}{\relarrow#3}{#1\depfunrelcolons#2}
\DeclarePairedDelimiterXPP\depfunrelpredrest[4]{}{[}{]}{\relarrow#4}{#1\depfunrelcolons#2\mid#3}

\DeclarePairedDelimiterXPP\depmonopred[3]{}{[}{]}{\monoarrow#3}{#1\depfunrelcolons#2}
\DeclarePairedDelimiterXPP\depmonopredrest[4]{}{[}{]}{\monoarrow#4}{#1\depfunrelcolons#2\mid#3}

\newcommand{\depfunmapcolons}{::}
\newcommand{\funmaparrow}{\rightarrow}
\DeclarePairedDelimiterXPP\depfunmap[3]{}{[}{]}{\funmaparrow#3}{#1\depfunmapcolons#2}

\newcommand{\funmap}[2]{#1\funmaparrow#2}

\DeclarePairedDelimiterXPP\halfgalpl[4]{}{(}{)}{\app#3\app#4}{#1\app{}_{\constfont{h}}{\unlhd}\app#2}
\DeclarePairedDelimiterXPP\halfgalpr[4]{}{(}{)}{\app#3\app#4}{#1\unlhd_{\constfont{h}}#2}
\DeclarePairedDelimiterXPP\galp[4]{}{(}{)}{\app#3\app#4}{#1\unlhd#2}
\newcommand{\galcsym}{\dashv}
\DeclarePairedDelimiterX\galc[2]{(}{)}{#1\galcsym#2}
\DeclarePairedDelimiterXPP\galcapp[4]{}{(}{)}{\app#3\app#4}{#1\galcsym#2}
\newcommand{\pregalcsym}{\galcsym_{\constfont{pre}}}
\DeclarePairedDelimiterXPP\pregalc[4]{}{(}{)}{\app#3\app#4}{#1\pregalcsym#2}
\newcommand{\galequivsym}{\equiv_{\constfont{G}}}
\DeclarePairedDelimiterXPP\galequiv[4]{}{(}{)}{\app#3\app#4}{#1\galequivsym#2}
\newcommand{\preequivsym}{\equiv_{\constfont{pre}}}
\DeclarePairedDelimiterXPP\preequiv[4]{}{(}{)}{\app#3\app#4}{#1\preequivsym#2}
\newcommand{\perequivsym}{\equiv_{\constfont{PER}}}
\DeclarePairedDelimiterXPP\perequiv[4]{}{(}{)}{\app#3\app#4}{#1\perequivsym#2}

\newcommand{\depfunleft}[2]{l_{2}\app#1\app#2}
\newcommand{\depfunright}[2]{r_{2}\app#1\app#2}

\newcommand{\reflrelsym}{\oplus}

\newcommand{\monorelarrow}{\Rrightarrow^{\reflrelsym}}
\DeclarePairedDelimiterXPP\monodepfunrel[4]{}{[}{]}{\monorelarrow#4}{#1\app#2\depfunrelcolons#3}

\newcommand{\monofunrel}[2]{#1\monorelarrow#2}
\newcommand{\nondepcase}[1]{#1}

\newcommand{\galrelsym}{\lessapprox}
\newcommand{\galrelconst}{\constfont{Galois}}
\newcommand{\galrel}[1]{\mathrel{{}_{#1}{\galrelsym}}}
\newcommand{\flipgalrel}[1]{\mathrel{{\galrelsym}{}_{#1}}}
\DeclarePairedDelimiterX\infgalrel[1]{(}{)}{{{}_{#1}{\galrelsym}}}
\DeclarePairedDelimiterX\flipinfgalrel[1]{(}{)}{{{\galrelsym}{}_{#1}}}

\newcommand{\natfuncmap}[1]{\constfont{map}_{#1}}
\newcommand{\natfuncrel}[1]{\constfont{rel}_{#1}}
\newcommand{\nargsc}[2]{#1_{1},\dotsc,#1_{#2}}
\newcommand{\nargs}[2]{#1_{1}\app\dotso\app#1_{#2}}
\newcommand{\nargsinf}[2]{\inflerel{#1_{1}}\app\dotso\app\inflerel{#1_{#2}}}
\newcommand{\nargsinfgalrel}[2]{\infgalrel{#1_{1}}\app\dotso\app\infgalrel{#1_{#2}}}

\DeclarePairedDelimiterXPP\infieldapp[1]{}{(}{)}{\app#1}{\infield\app#1}

\newcommand{\tyequivsym}{\simeq}
\DeclarePairedDelimiterXPP\tyequiv[4]{}{(}{)}{\app#3\app#4}{#1\tyequivsym#2}

\begin{document}

\title{Transport via Partial Galois Connections and Equivalences\arxivsubmittext{ (Extended Version)}{}}
\titlerunning{Transport via Partial Galois Connections and Equivalences}

\author{\anontext{Anon.\ Author}{Kevin Kappelmann\orcidID{0000-0003-1421-6497}}}
\institute{\anontext{Anon.\ Place}{Technical University of Munich, Boltzmannstrasse 3, Garching 85748, Germany,
\email{\anontext{anon.\ email}{kevin.kappelmann@tum.de}}}}

\maketitle

\begin{abstract}
Multiple types can represent the same concept.
For example,
lists and trees
can both represent
sets.
Unfortunately, this easily leads to incomplete libraries:
some set-operations may only be available on lists,
others only on trees.
Similarly,
subtypes and quotients are commonly used to
construct new type abstractions
in formal verification.
In such cases, one often wishes to reuse
operations on the representation type
for the new type abstraction, but to no avail:
the types are not the same.

To address these problems,
we present a new framework that
transports programs via equivalences.
Existing transport frameworks are either designed for
dependently typed, constructive proof assistants,
use univalence,
or are restricted to partial quotient types.
Our framework
\begin{enumerate*}
\item is designed for simple type theory,
\item generalises previous approaches working on partial quotient types, and
\item is based on standard mathematical concepts,
particularly Galois connections and equivalences.
\end{enumerate*}
We introduce the notions of partial Galois connection
and equivalence
and prove their closure properties under
(dependent) function relators,
(co)datatypes,
and compositions.
We formalised the framework in Isabelle/HOL and provide a prototype.\arxivsubmittext{\footnote{Non-peer reviewed, extended version of ``Transport via Partial Galois Connections and Equivalences'', 21st Asian Symposium on Programming Languages and Systems (APLAS), 2023~\cite{transportaplas}}}{}

\end{abstract}

\section{Introduction}\label{sec:introduction}

Computer scientists
often write programs and proofs in terms of representation types
but provide their libraries
in terms of different, though related, type abstractions.
For example,
the abstract type of finite sets
may be represented by the type of lists:
every finite set is related to every list
containing the same elements and, conversely,
every list is related to its set of elements.
As such, every function on lists respecting this relation
may be reused for a library on finite sets.
To be more explicit, consider the following example
in simple type theory:

\paragraph{A Simple Example}
Take the types of lists, $\listty{\alphaty}$,
and finite sets, $\fsetty{\alphaty}$.
There is a function $\tofset \holhasty \listty{\alphaty}\purefun\fsetty{\alphaty}$
that turns a list into its set of elements.
This allows us to define the relation
\makebox{$\listfsetrel\app xs\app s \define \tofset\app xs = s$}
that identifies lists and finite sets,
e.g.\
$\listfsetrel\app [1,2,3]\app \{1,2,3\}$
and
$\listfsetrel\app [3,1,2]\app \{1,2,3\}$.
Our goal is to use this identification to
transport programs between these two types.

For instance, take the function $\maxlist\app xs\define \foldr\app{\max}\app xs\app 0$ of type
\makebox{$\listty{\natty}\purefun\natty$}
that returns the maximum natural number contained in a list.
After some thinking, one recognises that $\maxlist$
respects the relation $\listfsetrel$ in the following sense:
if two lists correspond to the same set,
then applying $\maxlist$ to these lists returns equal results.
Formally,
\begin{equation}\label{eq:maxlist-listfsetrel}
\forall xs\app ys.\app \tofset\app xs = \tofset\app ys \holimplies \maxlist\app xs = \maxlist\app ys.
\end{equation}
Despite this insight,
we still cannot directly compute the maximum
of a finite set $s \holhasty \fsetty{\natty}$ using $\maxlist$;
the term $\maxlist\app s$ does not even typecheck (for good reasons).
But there is an indirect way
if we are also given an ``inverse'' of $\tofset$, call it
$\tolistfset \holhasty \fsetty{\alphaty}\purefun\listty{\alphaty}$,
that returns an arbitrary list containing the same elements as the given set.
The functions $\tofset$ and $\tolistfset$ form an equivalence between $\listty{\alphaty}$ and $\fsetty{\alphaty}$
that respects the relation $\listfsetrel$:
\begin{equation}\label{eq:tofset-tolistfset}
\forall xs.\app\listfsetrel\app xs\app (\tofset\app xs)
\qquad\text{and}\qquad
\forall s.\app\listfsetrel\app(\tolistfset\app s)\app s.
\end{equation}
Thanks to this equivalence,
we can compute the maximum of $s$ by simply transporting $s$ along the equivalence:
\begin{equation}\label{eq:maxfset-def}
\maxfset\app s \define \maxlist\app (\tolistfset\app s).
\end{equation}
The correctness of this transport is guaranteed by \labelcref{eq:maxlist-listfsetrel,eq:tofset-tolistfset,eq:maxfset-def}:
\begin{equation}
\forall xs\app s.\app\listfsetrel\app xs\app s \holimplies \maxlist\app xs = \maxfset\app s.
\end{equation}
We can now readily replace
any occurrence of $\maxfset\app s$
by
$\maxlist\app(\tolistfset\app s)$
and, vice versa,
any occurrence of $\maxlist\app xs$
by $\maxfset\app(\tofset\app xs)$.
This process can be extended
to many other functions,
such as $\map$, $\filter$, $\intersect$,
by introducing new terms $\mapfset$, $\filterfset$, $\intersectfset$
and proving their respectfulness theorems.
Indeed, it is a very repetitive task begging for
\makebox{automation}.

\paragraph{State of the Art}
There are various frameworks
to automate the transport
of terms along equivalences.
Most of them are designed for
dependently typed, constructive proof assistants
and are based on \emph{type equivalences}
\cite{depinterop1,depinterop2,proofrepair,univalparam1,univalparam2},
which play a central role in homotopy type theory.
In a nutshell, type equivalences are pairs of functions $f,g$
that are mutually inverse (i.e.\ $g\app (f\app x) = x$ and $f\app (g\app y) = y$)
together with a compatibility condition.
They cannot solve our problem since
$\tofset$ and $\tolistfset$ are not mutually inverse.

Angiuli et al.~\citet{univalparam3}
note and address
this issue
in Cubical Agda~\cite{cubicalagda}.
Essentially, they first quotient both types
and then obtain
a type equivalence between the quotiented types.
Their approach supports a restricted variant of
\emph{quasi-partial equivalence relations} \cite{quasipers}
but also uses univalence \cite{univalence},
which is unavailable in major proof assistants
like Isabelle/HOL~\cite{isabellehol}
and Lean~3~\cite{lean3}/Lean~4~\cite{lean4}.

Another existing framework is Isabelle's
\emph{Lifting package}~\cite{isabellelifting},
which transports terms via \emph{partial quotient types}:
\begin{definition}\label{def:partquot}
A partial quotient type $(T,Abs,Rep)$ is given by
a right-unique and right-total relation
$T$
and two functions $Abs$,
\makebox{$Rep$}
respecting $T$, that is
$T\app x\app y\holimplies Abs\app x = y$
and
$T\app (Rep\app y)\app y$,
for all  $x,y$.
\end{definition}
In fact, $\parenths{\listfsetrel,\tofset,\tolistfset}$ forms a partial quotient type.
The Lifting package can thus transport our list library to
finite sets\footnote{The Lifting package
is indeed used pervasively for such purposes.
At the time of writing, Isabelle/HOL and the \emph{Archive of Formal Proofs}
(\url{www.isa-afp.org}) contain more than 2800 invocations of the package.}.
However, the package also has its limitations:

\paragraph{Limitations of the Lifting Package}
Consider the previous example with one modification:
rather than transporting $\maxlist$ to finite sets,
we want to transport it to the type of (potentially infinite) sets, $\setty{\alphaty}$.
We cannot build a partial quotient type from $\listty{\alphaty}$
to $\setty{\alphaty}$ because the required relation
$T\holhasty\app \relty{\listty{\alphaty}}{\setty{\alphaty}}$
is not right-total (we can only relate finite sets to lists).
The Lifting package is stuck.
But in theory, we can (almost) repeat the previous process:
There is again a function $\toset \holhasty \listty{\alphaty}\purefun\setty{\alphaty}$.
We can define a relation
\makebox{$\listsetrel\app xs\app s \define \toset\app xs = s$}.
We can again prove that $\maxlist$ respects $\listsetrel$:
\begin{equation}
\forall xs\app ys.\app \toset\app xs = \toset\app ys \holimplies \maxlist\app xs = \maxlist\app ys.
\end{equation}
There is a function
$\tolistset \holhasty \setty{\alphaty}\purefun\listty{\alphaty}$,
and we obtain a \emph{partial} equivalence:
\begin{equation}\label{eq:toset-tolistset}
\forall xs.\app\listsetrel\app xs\app (\toset\app xs)
\qquad\text{and}\qquad
\forall s.\app\finite\app s\holimplies \app\listsetrel\app(\tolistset\app s)\app s.
\end{equation}
We can define the function
$\maxset\app s \define \maxlist\app (\tolistset\app s)$.
And we again obtain a correctness theorem:
$\forall xs\app s.\app\listsetrel\app xs\app s \holimplies \maxlist\app xs = \maxset\app s$.
While this process looks rather similar, there is one subtle change:
the second part of \cref{eq:toset-tolistset} only holds conditionally.
As a contribution of this paper,
we show that these conditions are not showstoppers,
and that we can transport via such partial equivalences in general.

Now one may argue that
we could still use partial quotient types to transport
from lists to sets:
First obtain a right-unique, right-total relation $T$ by
building a subtype of the target type.
Then transport to the new subtype
and then inject to the original type.
In spirit, this is close to the approach suggested by Angiuli et al.~\citet{univalparam3}.
But \anontext{the authors}{the author} \anontext{find}{finds} this unsatisfactory from a practical and a conceptual perspective:
From a practical perspective, it introduces unnecessary
subtypes to our theory.
And conceptually,
the process for sets and lists
was almost identical to the one for
finite sets and lists --
there was no detour via subtypes.

A second limitation of the Lifting package
is that it does not support \emph{inter-argument dependencies}.
For example,
take the types of natural numbers, $\natty$,
and integers, $\intty$.
We can construct a partial quotient type
$(\intnatrel,\tonat,\toint)$,
where
$\toint\holhasty \natty\purefun\intty$
is the standard embedding,
$\tonat\holhasty \intty\purefun\natty$ is its inverse (a partial function),
and
$\intnatrel\app i\app n \define i = \toint\app n$.
It then seems straightforward
to transport subtraction
$(-_\intty)\holhasty \intty\purefun\intty\purefun\intty$
from integers
to natural numbers in the following way:
\begin{equation}
n_1 -_\natty \app  n_2 \define \tonat\app\parenths[\big]{\toint\app n_1 -_\intty\toint\app n_2}.
\end{equation}
And of course, we expect a correctness theorem:
\begin{equation}
\forall i_1\app n_1\app i_2\app n_2.\app\intnatrel\app i_1\app n_1 \land \intnatrel\app i_2\app n_2\holimplies \intnatrel\app (i_1-_\intty i_2)\app (n_1-_\natty n_2).
\end{equation}
But alas, the theorem does not hold:
we need an extra dependency
between the arguments of the respective subtractions,
e.g.\ $i_1 \geq i_2$ or $n_1\geq n_2$.
Unfortunately,
the Lifting package's theory~\cite{isabellelifting}
cannot account for such dependencies,
and as such,
the transport attempt for $(-_\intty)$ fails.

In a similar way,
the list index operator
\makebox{$\inflistidx\holhasty \listty{\alphaty}\purefun\natty\purefun\alphaty$}
can only be transported
to the type of arrays for indices that are in bounds
(cf.\ \cref{sec:examples}, \cref{ex:transp_dep_fun_rel}).
While solutions for dependently typed
environments~\cite{depinterop1,depinterop2,proofrepair,univalparam1,univalparam2,univalparam3}
typically handle such examples by
encoding the dependencies in a type,
e.g.\
\makebox{$(xs\holhasty \listty{\alphaty})\purefun\{0,\dotsc,\length\app xs-1\}\purefun\alphaty$},
it is unclear how to support this
in a simply typed environment.
As a contribution of this paper,
we show how to account for such dependencies
with the help of \emph{dependent function relators}.

\paragraph{Contributions and Outline}
We introduce a new transport framework -- simply called \transport.
Our framework
\begin{enumerate*}
\item is applicable to simple type theory,
\item is richer than previous approaches working on partial quotient types, and
\item is based on standard mathematical notions, particularly Galois connections and equivalences.
\end{enumerate*}
In \cref{sec:essence-transport},
we distil the essence of what we expect when
we transport terms via equivalences.
The derived set of minimal expectations motivates us
to base our framework on Galois connections.

To meet these expectations,
we introduce the notion of partial Galois connections,
which generalise (standard) Galois connections and partial quotient types,
in \cref{sec:galconeqiv}.
We also introduce a generalisation of the well-known
function relator that allows for dependent relations in
\cref{sec:funrelmono}.

\cref{sec:closure} builds the technical core of the paper.
We derive closure conditions for
partial Galois connections and equivalences
as well as typical order properties (reflexivity, transitivity, etc.).
Specifically, we show closure properties
under (dependent) function relators,
relators for (co)datatypes, and composition.
All these results are novel and formalised in Isabelle/HOL.

Based on our theory,
we implemented a prototype for automated transports in Isabelle/HOL
and illustrate its usage in \cref{sec:examples}.
We conclude with related work in \cref{sec:relatedwork}
and future work in \cref{sec:conclusion}.

This article’s \arxivsubmittext{supplementary material\footnote{\url{https://www.isa-afp.org/entries/Transport.html}}}{extended version~\cite{transportarxiv}}
includes the formalisation and a guide linking
all definitions, results, and examples to their formal~counterpart in Isabelle/HOL.

\section{The Essence of Transport}\label{sec:essence-transport}

Existing frameworks, although beneficial in practical contexts,
are unapplicable to our introductory examples.
We hence first want to find \emph{the essence of transport}\footnote{
To avoid confusion,
our work is not about the $\constfont{transport}$ map
from homotopy type theory~\cite[Chapter~2]{hottbook}.
We focus on the general task of transporting a term $t$
to another term $t'$
along some notion of equivalence (not necessarily a type equivalence).}.
To find this essence, we have to answer the following question:
\begin{center}
\emph{\makebox{What are the minimum expectations when we transport terms via equivalences?}}
\end{center}
In this section, we argue that Galois connections are the right notion to cover this essence.
Let us examine prior work to identify some guiding principles.

\paragraph{Type Equivalences}
Much recent work is based on type equivalences~\cite{univalparam1,univalparam2,univalparam3,proofrepair,depinterop1,depinterop2}.
We denote a type equivalence between
$\alphaty$ and $\betaty$ with mutual inverses
$f\holhasty\alphaty\purefun\betaty$ and
$g\holhasty\betaty\purefun\alphaty$
by
$\tyequiv{\alphaty}{\betaty}{f}{g}$.
Then, on a high level, given a set of equivalences
$\tyequiv{\alphaty_i}{\betaty_i}{f_i}{g_i}$ for $1\leq i\leq n$
and two target types $\alphaty,\betaty$ that may include $\alphaty_i,\betaty_i$,
one tries to build an equivalence $\tyequiv{\alphaty}{\betaty}{f}{g}$.
Given a term $t\holhasty\alphaty$,
we can then define $t'\define f\app t$,
satisfying $t = g\app t'$.
Symmetrically, for a term $t'\holhasty\betaty$,
we can define $t\define g\app t'$,
satisfying $f\app t = t'$.
This situation is depicted in~\cref{fig:typeequivalence}.

\begin{figure}[t]
\begin{subfigure}[t]{0.49\textwidth}
  \centering
  
  \begin{tikzpicture}
  \pic (alpha) {typeshift={leftstyle}{$(=)$}{-0.5cm}};

  \pic (x1) at ($(alphatype) + (-0.3, 1.4)$) {typeequivclassleft};
  \node[left=\typeequivlabeldistance of x1element] {$t$};

  \pic (x2) at ($(x1element) + (\equivdistance, -\equivdistance)$) {typeequivclassleft};
  \pic (x3) at ($(x2element) + (-\equivdistance, -\equivdistance)$) {typeequivclassleft};
  \pic (x4) at ($(x3element) + (\equivdistance, -\equivdistance)$) {typeequivclassleft};

  \pic (beta) at ($(alphatype) + (\defaulttypedistance, 0)$) {typeshift={rightstyle}{{$(=)$}}{-0.5cm}};

  \pic (y1) at ($(x1element) + (\defaulttypedistance,0)$) {typeequivclassright};
  \node[right=\typeequivlabeldistance of y1element] {$t'$};
  \pic (y2) at ($(x2element) + (\defaulttypedistance,0)$) {typeequivclassright};
  \pic (y3) at ($(x3element) + (\defaulttypedistance,0)$) {typeequivclassright};
  \pic (y4) at ($(x4element) + (\defaulttypedistance,0)$) {typeequivclassright};

  \pic {typeequivedgepair={x1element}{y1element}{$f$}{$g$}};
  \pic {typeequivedgepair={x2element}{y2element}{}{}};
  \pic {typeequivedgepair={x3element}{y3element}{}{}};
  \pic {typeequivedgepair={x4element}{y4element}{}{}};

\end{tikzpicture}

  \caption{Example of a type equivalence. Left and right-hand side relation are restricted to be equality.}\label{fig:typeequivalence}
\end{subfigure}
\hfill
\begin{subfigure}[t]{0.49\textwidth}
  \centering
  
  \begin{tikzpicture}
  \pic (alpha) {typeshift={leftstyle}{{$(\approx)$}}{-0.5cm}};

  \node[element] (x1) at ($(alphatype) + (-0.3, 1.4)$) {};
  \node[left=\labeldistance of x1] {$t$};

  \node[element] (x2) at ($(x1) + (\equivdistance, -\equivdistance)$) {};
  \pic[leftstyle] {equivclass={x1}{x2}{1cm}};
  \pic (x3) at ($(x2) + (-\equivdistance, -\equivdistance)$) {typeequivclassleft};
  \node[element] (x4) at ($(x3element) + (\equivdistance, -\equivdistance)$) {};
  \node[element] (x5) at ($(x4) + (-\equivdistance, -0.5*\equivdistance)$) {};

  \pic (beta) at ($(alphatype) + (\defaulttypedistance, 0)$) {typeshift={rightstyle}{{$(=)$}}{-0.5cm}};

  \pic (y1) at ($(x1) + (\defaulttypedistance,0)$) {typeequivclassright};
  \node[right=\typeequivlabeldistance of y1element] {$t'$};

  \node[none] (y2) at ($(x2) + (\defaulttypedistance,0)$) {};
  \pic (y3) at ($(x3element) + (\defaulttypedistance,0)$) {typeequivclassright};

  \pic {leftedge={x1}{y1element}{$Abs$}};
  \pic {leftedge={x2}{y1element}{}};
  \draw[->] (y1element) edge[edgestyle,rightstyle,"$Rep$",pos=0.75] (x2);
  \pic {typeequivedgepair={x3element}{y3element}{}{}};

\end{tikzpicture}

  \caption{Example of a partial quotient type. The left relation can be an arbitrary partial equivalence relation. The right relation is restricted to be equality.}\label{fig:partquottype}
\end{subfigure}
\caption{Examples of equivalences used in prior work.
Types are drawn solid, black.
Transport functions are drawn dashed.
Each equivalence gives rise to a number of equivalence classes
on the left and right-hand side of the equivalence,
which are drawn dotted.
Arrows inside equivalence classes are omitted.
}\label{fig:typeequivalence_partquottype}
\Description[Prior work is restricted to equivalences with equality as a right relation]{Prior work supports type equivalences and partial quotient types. Both are restricted to equality for the right-hand side relation of the equivalence.}
\end{figure}

\paragraph{Partial Quotient Types}
The Lifting package~\cite{isabellelifting}
is based on partial quotient types $(T,Abs,Rep)$ (see \cref{def:partquot}).
Every partial quotient type
induces a relation $(\approx)\holhasty\relty{\alphaty}{\alphaty}$
that identifies values in $\alphaty$
that map to the same value in~$\betaty$:
\begin{equation}\label{eq:partequivinducedrel}
x_1 \approx x_2\define \indom\app T\app x_1\land Abs\app x_1 = Abs\app x_2.
\end{equation}
Given a set of partial quotient types
$(T_i\holhasty\relty{\alphaty_i}{\betaty_i},\allowbreak Abs_i,Rep_i)$ for $1\leq i\leq n$
and two target types $\alphaty,\betaty$ that may include $\alphaty_i,\betaty_i$,
the Lifting package tries to build
a partial quotient type
$(T\holhasty\relty{\alphaty}{\betaty},Abs,Rep)$.
Given a term $t$ in the domain of $(\approx)$,
we can then define $t'\define Abs\app t$, satisfying
$t\approx Rep\app t'$.
Symmetrically, for a term $t'\holhasty\betaty$,
we can define $t\define Rep\app t'$,
satisfying $Abs\app t = t'$.
This situation is depicted in~\cref{fig:partquottype}.

\paragraph{The Essence}
Abstracting from these approaches, we note some commonalities:
\begin{itemize}
\item As input, they take base equivalences,
which are then used to build more complex equivalences.
\item The equivalences include
a \emph{left transport function}
$l\holhasty \alphaty\purefun\betaty$
and
a \emph{right transport function}
$r\holhasty \betaty\purefun\alphaty$.
They can be used to move terms
from one side of the equivalence to a ``similar''
term on the other side of the equivalence.
\item Terms $t\holhasty\alphaty$ and $t'\holhasty\betaty$ that are ``similar'' stand in particular relations:
in the case of type equivalences, $t = r\app t'$ and $l\app t = t'$;
in the case of Lifting, $t \approx r\app t'$ and $l\app t = t'$.
More abstractly,
$L\app t\app (r\app t')$ and $R\app (l\app t)\app t'$
for some \emph{left relation}
\makebox{$L\holhasty\relty{\alphaty}{\alphaty}$}
and \emph{right relation}
$R\holhasty\relty{\betaty}{\betaty}$.%
\footnote{
The choice of
$L\app t\app (r\app t'), R\app (l\app t)\app t'$
may seem arbitrary -- why not pick
$L\app t\app (r\app t'), R\app t'\app (l\app t)$ instead?
In the end, the choice does not matter:
While the former leads us to (monotone) Galois connections,
the latter leads us to antitone Galois connections.
Using that $L,R$ form a Galois connection if and only if
$L,R^{-1}$ form an antitone Galois connection,
every result in this paper
can be transformed to its corresponding result on
antitone Galois connections
by an appropriate instantiation of the framework.
}
\item More generally,
$L$ and $R$
specify how terms ought to be related in $\alphaty$ and~$\betaty$
and determine which terms can be meaningfully transported using $l$ and $r$.
\item $L,R,l,r$ are compatible: if
terms are related on one side (e.g.\ $L\app t_1\app t_2$), their transports are related on the other side (e.g.\ $R\app (l\app t_1)\app (l\app t_2)$).
\end{itemize}
Based on these commonalities,
we can formulate six minimum expectations:
\begin{expect}
\item We want to specify how terms in $\alphaty$ and $\beta$
are related using relations $L,R$.
\item Transports should be possible by means of functions $l\holhasty\alphaty\purefun\betaty,r\holhasty\betaty\purefun\alphaty$.
\item\label{expect:closure}The notion of equivalence should be closed under common relators,
particularly those for functions and (co)datatypes.
\item\label{expect:mono}Terms related on one side have transports that are related on the other~side.
\item\label{expect:simtrans}Transporting a term should result in a term that is ``similar'' to its input.
\item\label{expect:reltrans}``Similar'' terms
$t\holhasty\alphaty$ and $t'\holhasty\betaty$
are related with each other's transports,
i.e.\
$L\app t\app (r\app t')$ and $R\app (l\app t)\app t'$.
\end{expect}
Applying~\cref{expect:reltrans}
to~\cref{expect:simtrans}
then yields the requirements
\begin{require}(2)
\item\label{require:inflat} $L\app t\app (r\app(l\app t))$,
\item\label{require:deflat} $R\app (l\app (r\app t'))\app t'$.
\end{require}
At this point,
one may notice
the similarity to \emph{Galois connections}.
A Galois connection between two preorders
\makebox{$\inflerel{L}$}
and $\inflerel{R}$
consists of two functions
$l$ and
$r$
such that
\begin{itemize}
\item $l$ is monotone, that is $x_1\lerel{L}x_2\holimplies l\app x_1\lerel{R}l\app x_2$ for all $x_1,x_2$,
\item $r$ is monotone, that is $y_1\lerel{R}y_2\holimplies r\app y_1\lerel{L}r\app y_2$ for all $y_1,y_2$, and
\item $x \lerel{L} r\app(l\app x)$ and $l\app (r\app y) \lerel{R} y$ for all $x,y$.\footnote{These two conditions are equivalent to requiring
$x \lerel{L}r\app y \holiff l\app x\lerel{R} y$ for all $x,y$.}
\end{itemize}
The final conditions correspond
to \cref{require:inflat,require:deflat} above,
while the monotonicity conditions on $l$ and $r$
correspond to \cref{expect:mono}.

\paragraph{Other Motivations}
A second motivation to base our framework
on Galois connections
comes from category theory.
There, an equivalence between two categories $L,R$ is given
by two functors
$l : L\to R$ and
$r : R\to L$
and two natural isomorphisms
$\unit : Id_L\to r\comp l$ and
$\counit : l\comp r \to Id_R$.
Applied to preorders $\inflerel{L},\inflerel{R}$ and monotone functions $l,r$,
this translates to the four conditions
\begin{condition}(4)
\item\label{cond:unit}$x\lerel{L} r\app(l\app x)$,
\item\label{cond:counit}$l\app (r\app y) \lerel{R} y$,
\item $r\app(l\app x)\lerel{L} x$,
\item $y \lerel{R} l\app (r\app y)$.
\end{condition}
A related categorical concept is that of an \emph{adjunction}.
When applied to preorders and monotone functions,
an adjunction is similar to an equivalence but is only required to satisfy
\cref{cond:unit,cond:counit}.
In fact, while Galois connections are not categorical equivalences,
they are adjunctions.
From this perspective,
a Galois connection can be seen as a weak form of an (order) equivalence.

A final motivation is the applicability and wide-spread use
of Galois connections.
They are fundamental in the closely related field
of abstract interpretation~\cite{abstractint1,abstractintbook},
where they are used to relate concrete to abstract domains.
Moreover, they are pervasive throughout mathematics.
In the words of Saunders Mac Lane:
\begin{quote}
\emph{The slogan is ``Adjoint functors arise everywhere''.}

\hfill (Categories for the Working Mathematician)
\end{quote}

We hope
our exposition convinced the reader that
Galois connections are
a suitable notion to cover the essence of transport.
The remaining challenges are
\begin{itemize}
\item to bring the notion of Galois connections to a partial world -- the relations
$L,R$
may only be defined on a subset of $\alphaty,\betaty$ -- and
\item to check the closure properties of our definitions under common relators.
\end{itemize}

\section{Partial Galois Connections, Equivalences, and Relators}\label{sec:galconequivrel}

In the previous section, we singled out Galois connections as a promising candidate for \transport.
Now we want to bring our ideas to the formal world of proof assistants.
In this section, we introduce the
required background theory for this endeavour.
In the following, we fix two relations
$L\holhasty\relty{\alphaty}{\alphaty}$,
$R\holhasty\relty{\betaty}{\betaty}$
and two functions
$l\holhasty\alphaty\purefun\betaty$,
$r\holhasty\betaty\purefun\alphaty$.

\subsection{(Order) Basics}

We work in a polymorphic, simple type theory~\cite{churchstt},
as employed,
for example, in Isabelle/HOL~\cite{isabellehol}.
In particular, our formalisation uses function extensionality.
We assume basic familiarity with Isabelle's syntax.
Here, we only recap the most important concepts for our work.
A complete list of definitions can be found in~\appendixref{sec:appendixorderbasics}.

A \emph{predicate on a type $\alphaty$} is a function of type
$\alphaty \purefun \bool$.
A \emph{relation on $\alphaty$ and $\betaty$} is a function of type
$\relty{\alphaty}{\betaty}$.
\emph{Composition of two relations $R,S$}
is defined as $(R\relcomp S)\app x\app y\define\exists z.\ R\app x\app z\land S\app z\app y$.
A \emph{relation $R$ is finer than a relation $S$}, written $R\le S$,
if $\forall x\app y.\app R\app x\app y\holimplies S\app x\app y$.
It will be convenient to interpret relations as infix operators.
For every relation $R$, we hence introduce an infix operator $\inflerel{R}\define R$,
that is $x \lerel{R} y \holiff R\app x\app y$.
We also write $\infgerel{R}\define \inv{\inflerel{R}}$.
The \emph{field predicate on a relation} is
defined as
$\infield\app R\app x\define \indom\app R\app x\lor \incodom\app R\app x$.

We use relativised versions of well-known order-theoretic concepts.
For example,
given a predicate $P$,
we define \emph{reflexivity on $P$ and $R$} as
$\reflon\app P\app R\define
\forall x.\app P\app x \holimplies R\app x\app x$.
We proceed analogously for other standard order-theoretic concepts, such as transitivity, preorders, etc. (see~\appendixref{sec:appendixorderbasics}).

\subsection{Function Relators and Monotonicity}\label{sec:funrelmono}

We introduce a generalisation of the well-known function relator (see e.g.~\cite{reynoldsparampoly}).
The slogan of the function relator is
``related functions map related inputs to related outputs''.
Our generalisation --
the \emph{dependent function relator} --
additionally allows its target relation to depend on both inputs:
\begin{equation}
\parenths[\big]{\depfunrel{x}{y}{R}{S}}\app f\app g \define
\forall x\app y.\app R\app x\app y\holimplies S\app(f\app x)\app(g\app y),
\end{equation}
where $x,y$ may occur freely in $S$.
The well-known \emph{(non-dependent) function relator} is given as a special case:
$\parenths{\funrel{R}{S}} \define \parenths[\big]{\depfunrel{\_}{\_}{R}{S}}$.
A function is \emph{monotone from $R$ to $S$} if it maps $R$-related inputs to $S$-related outputs:
\begin{equation}
\parenths[\big]{\depmono{x}{y}{R}{S}}\app f \define \parenths[\big]{\depfunrel{x}{y}{R}{S}}\app f\app f,
\end{equation}
where $x,y$ may occur freely in $S$.
A \emph{monotone function relator} is like a function relator but additionally requires its members to be monotone:
\begin{equation}
\begin{aligned}
\parenths[\big]{\monodepfunrel{x}{y}{R}{S}}\app f\app g \define &\parenths[\big]{\depfunrel{x}{y}{R}{S}}\app f\app g\\
                                                                &\land \parenths[\big]{\depmono{x}{y}{R}{S}}\app f
                                                                \land \parenths[\big]{\depmono{x}{y}{R}{S}}\app g,
\end{aligned}
\end{equation}
where $x,y$ may occur freely in $S$.
In some examples, we have to include conditionals in our relators.
For this, we define the \emph{relational if conditional}
$\relif\app B\app S\app x\app y \define B \holimplies S\app x\app y$
and set the following notation:
\begin{align}
\parenths[\big]{\depfunrelrest{x}{y}{R}{B}{S}}&\define
\parenths[\big]{\depfunrel{x}{y}{R}{\relif\app B\app S}},
\end{align}
where $x,y$ may occur freely in $B,S$.

\subsection{Galois Relator}\label{sec:galrel}

In \cref{expect:reltrans} of \cref{sec:essence-transport},
we noted
that
``similar'' terms
$t,t'$
are related with each other's transports,
i.e.\
$L\app t\app (r\app t')$ and $R\app (l\app t)\app t'$.
We now define this relation formally,
calling it the \emph{Galois relator}:
\begin{equation}
\galrelconst\app \inflerel{L}\app \inflerel{R}\app r\app x\app y\define
\incodom \inflerel{R}\app y \land x \lerel{L} r\app y
\end{equation}
When the parameters are clear from the context,
we will use the infix notation
\makebox{$\infgalrel{L}\define\galrelconst\app \inflerel{L}\app \inflerel{R}\app r$}.
It is easy to show that Galois relators
generalise the transport relations
of partial quotient types:
\begin{lemma}\label{lem:galrelpartquoteq}
For every partial quotient type $(T,l,r)$ with induced
left relation
$\inflerel{L}$,
we have
\makebox{$T = \galrelconst\app \inflerel{L}\app (=)\app r$}.
\end{lemma}

\subsection{Partial Galois Connections and Equivalences}\label{sec:galconeqiv}

In their standard form,
Galois connections are defined on preorders
$\inflerel{L},\inflerel{R}$,
where
every $x\holhasty\alphaty$
is in the domain of $\inflerel{L}$
and
every $y\holhasty\betaty$
is in the domain of $\inflerel{R}$.
But as we have seen,
this is not generally the case
when \makebox{transporting terms}.

We hence lift the notion of Galois connections to a partial setting.
We also do not assume any order axioms on $\inflerel{L},\inflerel{R}$
a priori but add them as needed.
In our formalisation,
we moreover break the concept
of Galois connections
down into smaller pieces
that, to our knowledge,
do not appear as such in the literature.
This allows us to obtain very precise results
when deriving the closure properties for our definitions
(\cref{sec:closure}).
But for reasons of brevity,
we only state the main definitions and results here.
Details can be found in \appendixref{sec:appendixgalconeqiv}.

The \emph{(partial) Galois property} is defined as:
\begin{equation}\label{eq:def_gal_prop}
\begin{aligned}
\galp[\big]{\inflerel{L}}{\inflerel{R}}{l}{r} \define
&\forall x\app y.\app \indom\app \inflerel{L}\app x\land \incodom \inflerel{R}\app y\holimplies{}\\
&(x \lerel{L} r\app y\holiff l\app x \lerel{R} y).
\end{aligned}
\end{equation}
If $l$ and $r$ are also monotone,
we obtain a \emph{(partial) Galois connection}:
\begin{equation}\label{eq:def_gal_conn}
\begin{aligned}
\galcapp[\big]{\inflerel{L}}{\inflerel{R}}{l}{r} \define
&\galp[\big]{\inflerel{L}}{\inflerel{R}}{l}{r} \\
&\land\parenths[\big]{\mono{\inflerel{L}}{\inflerel{R}}}\app l \land
\parenths[\big]{\mono{\inflerel{R}}{\inflerel{L}}}\app r.
\end{aligned}
\end{equation}
We omit the qualifier ``partial'' when referring to these
definitions, unless we want to avoid ambiguity.
An example Galois connection can be found in \cref{fig:galcon}.
\begin{figure}[t]
\begin{subfigure}[t]{0.49\textwidth}
  \centering
  
  \begin{tikzpicture}
  \pic (alpha) {typeshift={leftstyle}{$\inflerel{L}$}{-0.60cm}};

  \node[element] (x1) at ($(alphatype) + (-0.3, 1.4)$) {};
  \node[element] (x2) at ($(x1) + (\equivdistance, -\equivdistance)$) {};
  \node[element] (x1') at ($(x1)!0.5!(x2)$) {};
  \node[element] (x3) at ($(x2) + (-\equivdistance, -\equivdistance)$) {};
  \node[element] (x2') at ($(x1)!0.5!(x3)$) {};
  \node[element] (x4) at ($(x3) + (\equivdistance, -\equivdistance)$) {};
  \node[element] (x5) at ($(x4) + (-\equivdistance, -0.5*\equivdistance)$) {};

  \pic (beta) at ($(alphatype) + (\defaulttypedistance, 0)$) {typeshift={rightstyle}{{$\inflerel{R}$}}{-0.60cm}};

  \node[element] (y1) at ($(x1) + (\defaulttypedistance,0)$) {};
  \node[element] (y2) at ($(x2) + (\defaulttypedistance,0)$) {};
  \node[element] (y3) at ($(x3) + (\defaulttypedistance,0)$) {};
  \node[element] (y4) at ($(x4) + (\defaulttypedistance,0)$) {};
  \node[element] (y5) at ($(x5) + (\defaulttypedistance,0)$) {};
  \path let \p1 = (y1) in let \p2 = (y2) in
    node[element] (y1') at (\x2,\y1) {};
  \node[element] (y2') at ($(y2)!0.3!(y4)$) {};

  \draw[->,leftstyle,relstyle] (x1) to (x1');
  \draw[->,leftstyle,relstyle] (x1') to (x2);
  \draw[->,leftstyle,relstyle] (x2) to (x3);
  \draw[->,loop left,loopstyle,leftstyle,relstyle] (x1') to (x1');
  \draw[->,loop left,loopstyle,leftstyle,relstyle] (x3) to (x3);
  \draw[->,loop above,loopstyle,leftstyle,relstyle] (x5) to (x5);
  \draw[->,rightstyle,relstyle] (y1) to (y2);
  \draw[->,rightstyle,relstyle] (y2) to (y3);
  \draw[->,rightstyle,relstyle] (y3) to (y4);
  \draw[->,loop right,loopstyle,rightstyle,relstyle] (y1) to (y1);
  \draw[->,loop right,loopstyle,rightstyle,relstyle] (y3) to (y3);
  \draw[->,loop right,loopstyle,rightstyle,relstyle] (y5) to (y5);
  \pic {leftedge={x1}{y1}{$l$}};
  \pic {leftedge={x1'}{y1}{}};
  \pic {rightedge={y1}{x1'}{}};
  \pic {typeequivedgepair={x2}{y2}{}{}};
  \draw[->] (y4) edge[rightstyle,edgestyle,bend left=15] (x3);
  \pic {typeequivedgepair={x3}{y3}{}{}};
  \pic {typeequivedgepair={x5}{y5}{}{$r$}};

\end{tikzpicture}

  \caption{A partial Galois connection. Note that unlike in \cref{fig:typeequivalence_partquottype}, the relations may not decompose into equivalence classes.}\label{fig:galcon}
\end{subfigure}
\hfill
\begin{subfigure}[t]{0.49\textwidth}
  \centering
  
  \begin{tikzpicture}
  \pic (alpha) {typeshift={leftstyle}{$\inflerel{L}$}{-0.6cm}};

  \node[element] (x1) at ($(alphatype) + (-0.3, 1.4)$) {};
  \pic (x2) at ($(x1) + (\equivdistance, -\equivdistance)$) {typeequivclassleft};
  \node[element] (x1') at ($(x1)!0.5!(x2element)$) {};
  \pic[leftstyle] {equivclass={x1}{x1'}{0.5cm}};
  \pic (x3) at ($(x2element) + (-\equivdistance, -\equivdistance)$) {typeequivclassleft};
  \node[element] (x2element') at ($(x1)!0.5!(x3element)$) {};
  \node[element] (x4) at ($(x3element) + (\equivdistance, -\equivdistance)$) {};
  \pic (x5) at ($(x4) + (-\equivdistance, -0.5*\equivdistance)$) {typeequivclassleft};

  \pic (beta) at ($(alphatype) + (\defaulttypedistance, 0)$) {typeshift={rightstyle}{{$\inflerel{R}$}}{-0.6cm}};

  \pic (y1) at ($(x1) + (\defaulttypedistance,0)$) {typeequivclassright};
  \pic (y2) at ($(x2element) + (\defaulttypedistance,0)$) {typeequivclassright};
  \node[element] (y3) at ($(x3element) + (\defaulttypedistance,0)$) {};
  \node[element] (y4) at ($(x4) + (\defaulttypedistance,0)$) {};
  \pic[rightstyle] {equivclass={y3}{y4}{0.9cm}};
  \pic (y5) at ($(x5element) + (\defaulttypedistance,0)$) {typeequivclassright};
  \path let \p1 = (y1) in let \p2 = (y2) in
    node[element] (y1element') at (\x2,\y1) {};
  \node[element] (y2element') at ($(y2element)!0.3!(y4)$) {};

  \draw[->,leftstyle,relstyle] (x1') to (x2element);
  \draw[->,leftstyle,relstyle] (x2element) to (x3element);
  \draw[->,rightstyle,relstyle] (y1element) to (y2element);
  \draw[->,rightstyle,relstyle] (y2element) to (y3);
  \pic {leftedge={x1}{y1element}{$l$}};
  \pic {leftedge={x1'}{y1element}{}};
  \pic {rightedge={y1element}{x1'}{}};
  \pic {typeequivedgepair={x2element}{y2element}{}{}};
  \draw[->] (y4) edge[rightstyle,relstyle,edgestyle,bend left=15] (x3element);
  \pic {typeequivedgepair={x3element}{y3}{}{}};
  \pic {typeequivedgepair={x5element}{y5element}{}{$r$}};

\end{tikzpicture}

  \caption{A partial Galois equivalence. The relations decompose into
``strongly connected components'', drawn as dotted circles.
Any two members in such a component are connected.
These arrows are omitted.
}\label{fig:galequiv}
\end{subfigure}
\caption{Examples of partial equivalences as defined in \labelcref{eq:def_gal_conn},\labelcref{eq:def_gal_equiv}.
Types are drawn solid, black,
transport functions dashed,
and left and right relations dotted.}\label{fig:galcon_galequiv}
\Description[Examples of equivalences as defined in this work.]{Examples of a Galois connection and Galois equivalence as defined in this work. Galois connections may not decompose into equivalence classes. Galois equivalences essentially decompose into ``strongly connected components''.}
\end{figure}

As mentioned in \cref{sec:essence-transport},
Galois connections can be seen as a weak form of an equivalence.
Unfortunately,
they are not in general closed under compositions
(cf.\ \cref{sec:closurecomp}),
where we need a stronger form of an equivalence.
We can obtain a suitable strengthening
by requiring a two-sided Galois connection,
which we call a \emph{(partial) Galois equivalence}:
\begin{equation}\label{eq:def_gal_equiv}
\galequiv[\big]{\inflerel{L}}{\inflerel{R}}{l}{r} \define
\galcapp[\big]{\inflerel{L}}{\inflerel{R}}{l}{r}\land
\galcapp[\big]{\inflerel{R}}{\inflerel{L}}{r}{l}
\end{equation}
An example of a Galois equivalence can be found in \cref{fig:galequiv}.
It can be shown that Galois equivalences are,
under mild conditions,
equivalent to the traditional notion of
(partial) order equivalences~(see \appendixref{sec:appendixorderequivs}).

In practice,
the relations $\inflerel{L},\inflerel{R}$
are often preorders or partial equivalence relations (PERs).
Given some
$\galequiv[\big]{\inflerel{L}}{\inflerel{R}}{l}{r}$,
we hence introduce
the notations
$\preequiv[\big]{\inflerel{L}}{\inflerel{R}}{l}{r}$
and
$\perequiv[\big]{\inflerel{L}}{\inflerel{R}}{l}{r}$
in case both relations
$\inflerel{L},\inflerel{R}$
are preorders and
PERs on their domain,
respectively.
It is easy to show that Galois equivalences
generalise partial quotient types:
\begin{lemma}\label{lem:genpartquot}
$(T,l,r)$ is a partial quotient type with induced left relation $\inflerel{L}$
\textiff
$\perequiv[\big]{\inflerel{L}}{(=)}{l}{r}$.
\end{lemma}

\section{Closure Properties}\label{sec:closure}

We now explore the closure properties of
partial Galois connections and equivalences,
as well as standard order properties,
such as reflexivity and transitivity.
We will derive closure conditions
for the dependent function relator,
relators for (co)datatypes, and composition.
In each case, we will also derive conditions
under which the Galois relator
aligns with the context-dependent notion
\makebox{of ``similarity''}.

For reasons of brevity,
we only show that our framework is robust
under Galois equivalences on preorders and PERs here.
The results for Galois connections (and proof sketches)
can be found in \appendixref{sec:appendixclosuredepfunrel}.

\subsection{(Dependent) Function Relator}\label{sec:closuredepfunrel}

In the field of abstract interpretation,
it is well-known that Galois connections,
as usually defined in the literature,
are closed under the non-dependent, monotone function relator
(see for example~\cite{abstractint1}).
We generalise this result to partial Galois connections
and to dependent function relators.
\begin{remark}
The relations and functions
we use are often non-dependent in practice.
The following definitions
and theorems are considerably simpler in this case.
The reader hence might find instructive
to first consult the results for this special case
in~\appendixref{sec:appendixclosurefunrel}.
\end{remark}
\paragraph{The Setup}
In \cref{sec:introduction},
we highlighted the need of inter-argument dependencies
when transporting functions.
For example, we may only transport the index operator
$\inflistidx\holhasty \listty{\alphaty}\purefun\natty\purefun\alphaty$
if a given index is not out of bounds for a given list.
We can realise such dependencies with the help of the
dependent function relator from \cref{sec:funrelmono}.
For this, we fix the following variables:
\begin{align*}
\begin{split}
L_1 &\holhasty \alphaty_1 \purefun \alphaty_1 \purefun \bool,\\
R_1 &\holhasty \alphaty_2 \purefun \alphaty_2 \purefun \bool,\\
L_2 &\holhasty \alphaty_1 \purefun \alphaty_1 \purefun \betaty_1 \purefun \betaty_1 \purefun \bool,\\
R_2 &\holhasty \alphaty_2 \purefun \alphaty_2 \purefun \betaty_2 \purefun \betaty_2 \purefun \bool,
\end{split}
\begin{split}
l_1 &\holhasty \alphaty_1 \purefun \alphaty_2,\\
r_1 &\holhasty \alphaty_2 \purefun \alphaty_1,\\
l_2 &\holhasty \alphaty_2 \purefun \alphaty_1 \purefun \betaty_1 \purefun \betaty_2,\\
r_2 &\holhasty \alphaty_1 \purefun \alphaty_2 \purefun \betaty_2 \purefun \betaty_1.
\end{split}
\end{align*}
Each variable $L_2,R_2,l_2,r_2$ takes parameters from $\alphaty_1,\alphaty_2$.
These parameters enable the expression of inter-argument dependencies
(cf.~\cref{sec:examples}, \cref{ex:transp_dep_fun_rel}).
We hence call $L_2,R_2,l_2,r_2$ the \emph{dependent variables}.
Intuitively, we are in a situation where
\begin{enumerate}
\item we are given an equivalence between $\inflerel{L_1}$ and $\inflerel{R_1}$,
using $l_1$ and $r_1$,
\item whenever $x \galrel{L_1} x'$,
  we are given an equivalence between $\inflerel{L_2\app x\app (r_1\app x')}$ and $\inflerel{R_1\app (l_1\app x)\app x'}$,
using the transport functions $\depfunleft{x'}{x}$ and $\depfunright{x}{x'}$, and
\item we want to construct an equivalence for functions between\\
\makebox{$\parenths[\big]{\monodepfunrel{x_1}{x_2}{\inflerel{L_1}}{\inflerel{L_2\app x_1\app x_2}}}$}
and
\makebox{$\parenths[\big]{\monodepfunrel{x_1'}{x_2'}{\inflerel{R_1}}{\inflerel{R_2\app x_1'\app x_2'}}}$}.
\end{enumerate}
To define suitable transport functions,
we use the \emph{dependent function mapper}:
\begin{equation}
\parenths[\big]{\depfunmap{x}{f}{g}}\app h\app x \define g\app (f\app x)\app (h\app (f\app x)),
\end{equation}
where $x$ may occur freely in $g$.
We can now define the target relations and transport functions:
\begin{align}\label{eq:closuredepfunreldef}
\begin{split}
L &\define
\parenths[\big]{\monodepfunrel{x_1}{x_2}{\inflerel{L_1}}{\inflerel{L_2\app x_1\app x_2}}},\\
R &\define
\parenths[\big]{\monodepfunrel{x_1'}{x_2'}{\inflerel{R_1}}{\inflerel{R_2\app x_1'\app x_2'}}},
\end{split}
\begin{split}
l &\define
\parenths[\big]{\depfunmap{x'}{r_1}{l_2\app x'}},\\
r &\define
\parenths[\big]{\depfunmap{x}{l_1}{r_2\app x}}.
\end{split}
\end{align}
In particular,
$l\app f\app x' = \depfunleft{x'}{(r_1\app x')}\app \parenths[\big]{f\app (r_1\app x')}$
and
$r\app g\app x = \depfunright{x}{(l_1\app x)}\app \parenths[\big]{g\app (l_1\app x)}$.

\paragraph{Closure Theorems}
Checking the closure of order-theoretic concepts,
such as reflexivity, transitivity, and symmetry,
is fairly straightforward.
Verifying the closure of Galois connections and equivalences, however,
is nuanced,
requiring careful alignment of the dependent variables' parameters.
These alignments require the following \emph{monotonicity conditions},
which, broadly speaking, say that
\begin{enumerate*}[label=(\arabic*)]
\item $L_2,R_2$ are antimonotone in their first and monotone in their second parameter, and
\item $l_2, r_2$ are monotone in both parameters:
\end{enumerate*}
\begin{monocondition}
\item\label{asm:depfunrel_monoleft2_galequiv}%
If $x_1 \lerel{L_1} x_2 \lerel{L_1} x_3 \lerel{L_1} x_4$ then $\inflerel{L_2\app x_2\app x_3} \le \inflerel{L_2\app x_1\app x_4}$.
\item\label{asm:depfunrel_monoright2_galequiv}%
If $x_1' \lerel{R_1} x_2' \lerel{R_1} x_3' \lerel{R_1} x_4'$ then $\inflerel{R_2\app x_2'\app x_3'} \le \inflerel{R_2\app x_1'\app x_4'}$.
\item\label{asm:depfunrel_monol2_galequiv}%
If $x_1 \lerel{L_1} x_2 \galrel{L_1} x_1' \lerel{R_1} x_2'$ and $\infield\app\inflerel{L_2\app x_1\app (r_1\app x_2')}\app y$ then\\
$\parenths[\big]{\depfunleft{x_1'}{x_1}\app y} \lerel{R_2\app (l_1\app x_1)\app x_2'} \parenths[\big]{\depfunleft{x_2'}{x_2}\app y}$.
\item\label{asm:depfunrel_monor2_galequiv}%
If $x_1 \lerel{L_1} x_2 \galrel{L_1} x_1' \lerel{R_1} x_2'$ and $\infield\app\inflerel{R_2\app (l_1\app x_1)\app x_2'}\app y'$ then\\
$\parenths[\big]{\depfunright{x_1}{x_1'}\app y'} \lerel{L_2\app x_1\app (r_1\app x_2')} \parenths[\big]{\depfunright{x_2}{x_2'}\app y'}$.
\end{monocondition}
We are now ready to state our main result
for Galois equivalences on preorders and PERs.
The result for Galois connections (and a proof sketch)
can be found in \appendixref{sec:appendixclosuredepfunrel}.
All other results can be found in our formalisation.
\begin{theorem}\label{thm:depfunrel_galequiv}
Let $\star \in\{\preequivsym,\perequivsym\}$
and assume
\begin{assumes}(1)
\item$\parenths[\big]{\inflerel{L_1}\star\inflerel{R_1}}\app l\app r$,
\item\label{asm:depfunrel_galequiv2_galequiv}if $\app x \galrel{L_1} x'$ then $\parenths[\big]{\inflerel{L_2\app x\app (r_1\app x')}\star\inflerel{R_2\app (l_1\app x)\app x'}}\app (\depfunleft{x'}{x})\app (\depfunright{x}{x'})$,
\item \cref{asm:depfunrel_monoleft2_galequiv,asm:depfunrel_monoright2_galequiv,asm:depfunrel_monol2_galequiv,asm:depfunrel_monor2_galequiv}.
\end{assumes}
Then
\makebox{$\parenths[\big]{\inflerel{L}\star\inflerel{R}}\app l\app r$}.
\end{theorem}

\paragraph{``Similarity''}
Given the closure theorem,
we can readily transport a function $f$ from
$\inflerel{L}$
to a function $g$ in
$\inflerel{R}$.
Due to \cref{expect:mono,expect:reltrans},
we also know that $f \galrel{L} g$,
that is
$\parenths[\big]{\monodepfunrel{x_1}{x_2}{\inflerel{L_1}}{\inflerel{L_2\app x_1\app x_2}}}\app f\app (r\app g)$
and
\makebox{$\parenths[\big]{\monodepfunrel{x_1'}{x_2'}{\inflerel{R_1}}{\inflerel{R_2\app x_1'\app x_2'}}}\app (l\app f)\app r$}.
But arguably, this is not quite enough:

Remember the slogan of the function relator:
``related functions map related inputs to related outputs''.
We know how to relate terms between
$\inflerel{L_1}$ and
$\inflerel{R_1}$:
we can use $\infgalrel{L_1}$.
Whenever $x \galrel{L_1} x'$,
we also know how to relate terms between
$\inflerel{L_2\app x\app (r_1\app x')}$ and
$\inflerel{R_2\app (l_1 \app x)\app x'}$:
we can use
\begin{equation}
\infgalrel{L_2\app x\app x'}\define \galrelconst\app \inflerel{L_2\app x\app (r_1\app x')}\app\inflerel{R_2\app (l_1\app x)\app x'}\app (\depfunright{x}{x'}).
\end{equation}
So when we say that ``$f$ and $g$ are similar'',
we may actually desire that
\begin{equation}
\parenths[\big]{\depfunrel[\big]{x}{x'}{\infgalrel{L_1}}{\infgalrel{L_2\app x\app x'}}}\app f\app g.
\end{equation}
The following theorem answers when
$\infgalrel{L}$ aligns with this definition of similarity for preordered Galois equivalences.
Preciser results can be found in
\appendixref{sec:appendixclosuredepfunrel} and the formalisation.
\begin{theorem}\label{thm:depfunrel_galreleq}
Assume
\begin{assumes}
\item$\preequiv[\big]{\inflerel{L_1}}{\inflerel{R_1}}{l_1}{r_1}$,
\item if $\app x \galrel{L_1} x'$ then $\preequiv[\big]{\inflerel{L_2\app x\app (r_1\app x')}}{\inflerel{R_2\app (l_1\app x)\app x'}}{(\depfunleft{x'}{x})}{(\depfunright{x}{x'})}$,
\item \cref{asm:depfunrel_monoleft2_galequiv,asm:depfunrel_monor2_galequiv},
\item $\indom\app \inflerel{L}\app f$, and $\incodom\app \inflerel{R}\app g$.
\end{assumes}
Then
$f \galrel{L} g
\holiff
\parenths[\big]{\depfunrel{x}{x'}{\infgalrel{L_1}}{\infgalrel{L_2\app x\app x'}}}f\app g$.
\end{theorem}

\subsection{(Co)datatypes}

Different proof assistants ground (co)datatypes in
different ways.
For instance, Coq and Lean introduce them axiomatically,
whereas Isabelle/HOL proves their existence
using the theory of \emph{bounded natural functors}~\cite{bnf}.
As our formalisation takes place in Isabelle/HOL,
we use the latter theory.
Nonetheless,
the results presented in this section
are relatively straightforward
and can likely be adapted to
other ``reasonable'' definitions of (co)datatypes.

In this section, we derive closure properties
for arbitrary \emph{natural functors}.
A natural functor is a bounded natural functor without cardinality constraints.
The exact axioms
can be found elsewhere~\cite{bnf}.
For our purposes,
it suffices to say that natural functors are equipped with a \emph{mapper} and a \emph{relator}.
More precisely, for every $n$-ary natural functor $(\nargsc{\alphaty}{n})\app F$,
there are two functions:
\begin{align*}
  \natfuncmap{F} \holhasty\; &(\alphaty_1\purefun\betaty_1)\purefun\dotsb\purefun(\alphaty_n\purefun\betaty_n)\purefun(\nargsc{\alphaty}{n})\app F\purefun (\nargsc{\betaty}{n})\app F\\
  \natfuncrel{F} \holhasty\; &(\relty{\alphaty_1}{\betaty_1})\purefun\dotsb\purefun(\relty{\alphaty_n}{\betaty_n})\purefun{}\\
&\relty{(\nargsc{\alphaty}{n})\app F}{(\nargsc{\betaty}{n})\app F}
\end{align*}
The former lifts functions on the functor's type arguments
to the functorial structure,
the latter lifts relations on the functor's type arguments
to the functorial structure.
Using the mapper and relator,
it is straightforward
to define appropriate target relations and
transport functions.
First we fix the following variables for $1\leq i\leq n$:
\begin{displaymath}
L_i \holhasty \alphaty_i \purefun \alphaty_i \purefun \bool,\quad
l_i \holhasty \alphaty_i \purefun \betaty_i,\quad
R_i \holhasty \betaty_i \purefun \betaty_i \purefun \bool,\quad
r_i \holhasty \betaty_i \purefun \alphaty_i.
\end{displaymath}
Then we define the new target relations and transport functions as follows:
\begin{align}
\begin{split}
L &\define
\natfuncrel{F}\app\nargsinf{L}{n},\\
R &\define \natfuncrel{F}\app\nargsinf{R}{n},
\end{split}
\begin{split}
l &\define \natfuncmap{F}\app\nargs{l}{n},\\
r &\define \natfuncmap{F}\app\nargs{r}{n}.
\end{split}
\end{align}
The closure properties follow without any difficulty:
\begin{theorem}
Let
$\star \in\{\galcsym,\galequivsym,\preequivsym,\perequivsym\}$
and assume $\parenths[\big]{\inflerel{L_i}\star\inflerel{R_i}}\app l_i\app r_i$ for $1\leq i\leq n$.
Then \makebox{$\parenths[\big]{\inflerel{L}\star\inflerel{R}}\app l\app r$}.
\end{theorem}
As in the previous section,
we can ponder whether the relation $\infgalrel{L}$
adequately captures our desired notion of ``similarity''.
Again, we already know how to relate terms between
$\inflerel{L_i}$ and $\inflerel{R_i}$ for $1\leq i\leq n$:
we can use $\infgalrel{L_i}$.
We also know how to relate two functors:
we can use $\natfuncrel{F}$.
We thus may desire that ``$t$ and $t'$ are similar'' when
$\natfuncrel{F}\app \nargsinfgalrel{L}{n}\app t\app t'$.
It is easy to show that $\infgalrel{L}$ aligns with this desire:
\begin{theorem}\label{thm:natfunc_galreleq}
$\infgalrel{L} = \natfuncrel{F}\app \nargsinfgalrel{L}{n}$.
\end{theorem}
Proof details for this section can be found in our formalisation.
The formalisation includes tactic scripts that
are applicable to functors of arbitrary arity.
Integrating them into Isabelle/HOL's datatype package
is left as future work.

\subsection{Compositions}\label{sec:closurecomp}

It is well-known that Galois connections,
as defined in the literature,
are closed under composition
in the following sense:
given Galois connections between $\inflerel{L_1},\inflerel{R_1}$
and $\inflerel{L_2},\inflerel{R_2}$ with $\inflerel{R_1}=\inflerel{L_2}$,
we can build a Galois connection between
$\inflerel{L_1},\inflerel{R_2}$.
This result readily generalises
to our partial setting (see \appendixref{sec:appendixclosurecompcoincide}).
However, $\inflerel{R_1}$ and $\inflerel{L_2}$
usually do not coincide in our context.
We need a more general result.

\paragraph{The Setup}
Our goal is to define a notion of composition that works
even if
$\inflerel{R_1}$ and $\inflerel{L_2}$ do not coincide.
For this, we fix the variables
\begin{align*}
&L_1 \holhasty \alphaty \purefun \alphaty \purefun \bool,\quad
&l_1 \holhasty \alphaty \purefun \betaty,\quad
&R_1 \holhasty \betaty \purefun \betaty \purefun \bool,\quad
&r_1 \holhasty \betaty \purefun \alphaty,\\
&L_2 \holhasty \betaty \purefun \betaty \purefun \bool,\quad
&l_2 \holhasty \betaty \purefun \gammaty,\quad
&R_2 \holhasty \gammaty \purefun \gammaty \purefun \bool,\quad
&r_2 \holhasty \gammaty \purefun \betaty.
\end{align*}
Intuitively, we are in a situation where
\begin{enumerate}
\item we are given an equivalence between $\inflerel{L_1}$ and $\inflerel{R_1}$, using
$l_1$ and $r_1$,
\item we are given an equivalence between $\inflerel{L_2}$ and $\inflerel{R_2}$, using
$l_2$ and $r_2$, and
\item we want to construct an equivalence with
transport functions $l_2\comp l_1$ and $r_1\comp r_2$
between those parts of $\inflerel{L_1}$ and $\inflerel{R_2}$
that can be made ``compatible'' with respect to these functions.
This particularly means that we can apply the transport functions
on these parts without leaving the domains of the input equivalences.
\end{enumerate}
The question is: how do we find those parts and how can we make them compatible?
The solution we propose is inspired by and generalises
the approach of Huffman and Kun{\v{c}}ar~\citet{isabellelifting}.
We provide details and intuitions for the constructions in
\appendixref{sec:appendixclosurecomp}.
The resulting target relations and transport functions are defined as follows (where $\infgalrel{R_i}\define\galrelconst\app \inflerel{R_i}\app \inflerel{L_i}\app l_i$):
\begin{align}
\begin{split}
L &\define
\infgalrel{L_1} \relcomp \inflerel{L_2} \relcomp \infgalrel{R_1},\\
R &\define
\infgalrel{R_2} \relcomp \inflerel{R_1} \relcomp \infgalrel{L_2},
\end{split}
\begin{split}
l &\define l_2\comp l_1,\\
r &\define r_1\comp r_2.
\end{split}
\end{align}

\paragraph{Closure Theorems}

Again, we only state our main result
for Galois equivalences on preorders and PERs.
Preciser results can be found in \appendixref{sec:appendixclosurecomp} (including a proof sketch)
and in our formalisation.
\begin{theorem}\label{thm:comp_galequiv}
Let $\star \in\{\preequivsym,\perequivsym\}$ and assume
\begin{assumes}(2)
\item$\forall i\in\{1,2\}.\;\parenths[\big]{\inflerel{L_i}\star\inflerel{R_i}}\app l_i\app r_i$,
\item $\parenths[\big]{\inflerel{R_1} \relcomp \inflerel{L_2}} = \parenths[\big]{\inflerel{L_2} \relcomp \inflerel{R_1}}$.
\end{assumes}
Then \makebox{$\parenths[\big]{\inflerel{L}\star\inflerel{R}}\app l\app r$}.
\end{theorem}

\paragraph{``Similarity''}
For a final time,
we can ponder whether the relation $\infgalrel{L}$
is sufficient to capture our desired notion of ``similarity'':
Again, we already know how to relate terms between
$\inflerel{L_i}$ and $\inflerel{R_i}$ for $i\in\{1,2\}$:
we can use $\infgalrel{L_i}$.
We also have a natural way to combine these relations,
namely composition.
We thus may desire that ``$t$ and $t'$ are similar'' when
$\parenths[\big]{\infgalrel{L_1}\relcomp\infgalrel{L_2}}t\app t'$.
The next theorem answers when
$\infgalrel{L}$ aligns with this desire
for Galois equivalences.
Preciser results can be found in
\appendixref{sec:appendixclosurecomp}
and \makebox{the formalisation}.
\begin{theorem}\label{thm:comp_galreleq}
Assume
\begin{assumes}(2)
\item$\forall i\in\{1,2\}.\;\preequiv[\big]{\inflerel{L_i}}{\inflerel{R_i}}{l_i}{r_i}$,
\item $\parenths[\big]{\inflerel{R_1} \relcomp \inflerel{L_2}} = \parenths[\big]{\inflerel{L_2} \relcomp \inflerel{R_1}}$.
\end{assumes}
Then
$\infgalrel{L}=\parenths[\big]{\infgalrel{L_1}\relcomp\infgalrel{L_2}}$.
\end{theorem}

\section{Application Examples}\label{sec:examples}

As all our results are formalised in Isabelle/HOL,
we can directly use them to manually transport terms in said environment.
But that would be rather tiresome.
We thus implemented a prototype in Isabelle/ML to
\makebox{automate transports}.

\paragraph{The Prototype}
The method $\transporttermprover$
uses registered base equivalences,
along with the closure theorems from \cref{sec:closure},
to construct more complex equivalences.
The prototype is currently
restricted to equivalences on partial equivalence relations (PERs) for pragmatic reasons:
their closure theorems have fewer
assumptions and are hence simpler to apply.
Providing automation for weaker equivalences is future work.
The current prototype also does not build
composition closures (\cref{sec:closurecomp})
and automates only a fragment of
dependent function relators for simplicity reasons.
Again, these extensions are future work.

The prototype provides a command $\transportterm$.
As input, it takes a term $t\holhasty\alphaty$ (the term to be transported)
and two optional target relations
$L\holhasty\relty{\alphaty}{\alphaty}$,
$R\holhasty\relty{\betaty}{\betaty}$.
This is unlike other transport frameworks~\cite{univalparam1,proofrepair,depinterop1,isabellelifting},
which only take the term $t\holhasty \alphaty$ and a target type $\betaty$.
This design decision is crucial since we
can neither assume a unique correspondence
between types and target relations in practice
(cf.~\cref{ex:isaset}),
nor can we express dependencies in types,
but we express them using dependent relators
(cf. \cref{ex:transp_dep_fun_rel}).
The command then opens two goals.
The first one asks for an equivalence
$\perequiv[\big]{\inflerel{L}}{\inflerel{R}}{l}{r}$,
the second one for a proof that $\indom\app \inflerel{L}\app t$.
On success, it registers a new term $t'$ and a theorem that $t\galrel{L} t'$.
It also registers a second theorem where the relator $\infgalrel{L}$ has been rewritten
to its desired form
as described in \cref{thm:depfunrel_galreleq,thm:natfunc_galreleq,thm:comp_galreleq}.

The following examples are best explored interactively in our formalisation.
We define the \emph{restricted equality relation on predicates} as
$x =_P y\define P\app x\land x = y$
and
the \emph{restricted equality relation on sets} as
$x =_S y\define x\in S\land x = y$.
\begin{example}
It is easy to transport the list and set examples from \cref{sec:introduction}.
We just have to prove the equivalence between
$\listfsetrelL\app xs\app xs'\define \listfsetrel\app xs\app (\tofset\app xs')$
and
$(=) \holhasty \relty{\fsetty{\natty}}{\fsetty{\natty}}$
and invoke our prototype on $\maxlist$:
\begin{codealign}
&\lemmacmd\app\perintrotag{:}\app \perequiv{\listfsetrelL}{(=)}{\tofset}{\tolistfset}\\
&\transportterm\app \maxfset \holhasty \fsetty{\natty}\purefun \natty\app\wherecmd\app \transpparameq{t}\maxlist\app
\bycmd\app\transporttermprover
\end{codealign}
The $\perintrotag$ tag
is used by $\transporttermprover$
to discharge the closure theorems' side conditions.
$\transportterm$ registers the theorem
$\parenths[\big]{\funrel{\listfsetrel}{(=)}}\app \maxlist\app \maxfset$
and the definition
$\maxfset\app s\define \maxlist\app (\tolistfset\app s)$
as a result.
We can also readily transport in the opposite direction
or use $\settyconst$s rather than $\fsettyconst$s
if we
define $\listsetrelL\app xs\app xs'\define \listsetrel\app xs\app (\toset\app xs')$:
\begin{codealign}
&\transportterm\app \maxlist' \holhasty \listty{\natty}\purefun \natty\app\wherecmd\app \transpparameq{t} \maxfset\app
\bycmd\app\transporttermprover\\
&\lemmacmd\app\perintrotag{:}\app \perequiv{\listsetrelL}{(=_\finite)}{\toset}{\tolistset}\\
&\transportterm\app \maxset \holhasty \setty{\natty}\purefun \natty\app\wherecmd\app \transpparameq{t} \maxlist\app
\bycmd\app\transporttermprover
\end{codealign}
\end{example}
\begin{example}\label{ex:transp_dep_fun_rel}
As noted in \cref{sec:introduction}, transporting subtractions
$i_1 -_\intty i_2$ from $\intty$ to $\natty$
requires a dependency $i_1 \geq i_2$.
We model this dependency using dependent function relators.
We first define $\intspos\define (=_{(\le) 0})$ and then
proceed as usual:
\begin{codealign}
&\lemmacmd\app\perintrotag{:}\app \perequiv{\intspos}{(=)}{\tonat}{\toint}\\
&\transportterm\app (-_\natty)\holhasty \natty\purefun\natty\purefun\natty \app\wherecmd\app \transpparameq{t} (-_\intty)\\
&\quad\andcmd\app \transpparameq{L} \parenths[\big]{\depfunrel{i_1}{\_}{\intspos}{\depfunrelrest{i_2}{\_}{\intspos}{i_1 \geq i_2}{\intspos}}}\\
&\quad\andcmd\app \transpparameq{R} \parenths[\big]{\depfunrel{n_1}{\_}{(=)}{\depfunrelrest{n_2}{\_}{(=)}{n_1 \geq n_2}{(=)}}}\app\bycmd\app\transporttermprover
\end{codealign}
Similarly, operations on datatypes may only conditionally be transportable.
For example, we may only transport the index operator
$\inflistidx\holhasty \listty{\alphaty}\purefun\natty\purefun\alphaty$
to the type of immutable arrays ($\iarrayty{\alphaty}$)
if the index is not out of bounds.
In the following, let $S$ be an arbitrary partial equivalence relation:
\begin{codealign}
&\lemmacmd\app\perintrotag{:}\app\perequiv{\listrel\app S}{\iarrayrel\app S}{\toiarray}{\tolistiarray}\\
&\transportterm\app \iarrayindex\holhasty \iarrayty{\alphaty}\purefun\natty\purefun\alphaty\app\wherecmd\app \transpparameq{t} \inflistidx\\
&\quad\andcmd\app \transpparameq{L} \parenths[\big]{\depfunrel{xs}{\_}{\listrel\app S}{\depfunrelrest{i}{\_}{(=)}{i < \length\app xs}{S}}}\\
&\quad\andcmd\app \transpparameq{R} \parenths[\big]{\depfunrel{arr}{\_}{\iarrayrel\app S}{\depfunrelrest{i}{\_}{(=)}{i < \iarraylength\app arr}{S}}}\\
&\quad\bycmd\app\transporttermprover
\end{codealign}
\end{example}
\begin{example}\label{ex:isaset}
Isabelle/Set~\cite{isaset}
is a
set-theoretic environment in Isabelle/HOL.
Its type of sets is called $\mysetty$.
Isabelle/Set provides a \emph{set-extension} mechanism:
As input, it takes two sets $A\holhasty\mysetty$ and $B\holhasty\mysetty$
and an injection from $A$ to $B$.
It then creates a new set $B'\supseteq A$ together with a bijection
between $B$ and $B'$
with mutual inverses $l,r \holhasty\mysetty\purefun\mysetty$.
This mechanism is used to enforce subset relationships.
For instance, it first uses a construction of the integers $\intty\holhasty\mysetty$
where $\natty\not\subseteq\intty$.
It then uses the set-extension mechanism to create a copy $\intty'\supseteq \natty$ with inverses $l,r$.
Doing so necessitates a manual transport of all definitions from
$\intty$ to $\inttycopy$.
Using \transport,
it is possible to automate this process:
\begin{codealign}
&\lemmacmd\app\perintrotag{:}\app \perequiv{(=_\intty)}{(=_\inttycopy)}{l}{r}\\
&\transportterm\app (+_\inttycopy)\app\wherecmd\app \transpparameq{t} (+_\intty)
\app\andcmd\app \transpparameq{L} \parenths[\big]{\funrel{(=_\intty)}{\funrel{(=_\intty)}{(=_\intty)}}}\\
&\quad\app\andcmd\app \transpparameq{R} \parenths[\big]{\funrel{(=_\inttycopy)}{\funrel{(=_\inttycopy)}{(=_\inttycopy)}}}
\app\bycmd\app\transporttermprover\\
&\transportterm\app (-_\inttycopy)\app\wherecmd\app \transpparameq{t} (-_\intty)
\app\andcmd\app \transpparameq{L} \parenths[\big]{\funrel{(=_\intty)}{\funrel{(=_\intty)}{(=_\intty)}}}\\
&\quad\app\andcmd\app \transpparameq{R} \parenths[\big]{\funrel{(=_\inttycopy)}{\funrel{(=_\inttycopy)}{(=_\inttycopy)}}}
\app\bycmd\app\transporttermprover
\end{codealign}
Note that all constants $(+_\intty),(+_\inttycopy),(-_\intty),(-_\inttycopy)$
are of the same type $\mysetty\purefun\mysetty\purefun\mysetty$.
This stresses the point that users
must be able to specify target relations and not just target types.
\end{example}

\section{Related Work}\label{sec:relatedwork}

\paragraph{Transport in Proof Assistants}
Our work was chiefly inspired by Isabelle's
Lifting package~\cite{isabellelifting,kuncarthesis},
which transports terms via partial quotient types.
All closure theorems in this work
generalise the ones in~\cite{isabellelifting}.
Besides this source of inspiration,
the theory of automated transports
has seen prolific work in recent years:

Tabareau et al.~\citet{univalparam2}
proved a strengthened relational parametricity result,
called \emph{univalent parametricity},
for the Calculus of Inductive Constructions.
Their approach ensures
that all relations are compatible with type equivalences.
One can then use univalence~\cite{univalence} to seamlessly transport
terms between related types.
The framework is implemented using Coq's
typeclass mechanism~\cite{coqtypeclass}.

Tabareau et al.~\citet{univalparam1} extended their
work to integrate what they call
\emph{``white-box transports''}.
White-box transports
structurally rewrite a term $t$ to $t'$ using
user-specified correspondences.
In contrast,
\emph{``black-box transports''}
transport $t$ without looking at its syntactic structure.
For instance,
given an equivalence between unary and binary numbers
$\tyequiv{\natty}{\binty}{l}{r}$,
black-box transporting the term $0 +_{\natty} 0$ results
in $l\app (0 +_{\natty} 0)$.
In contrast,
given correspondences
between the functions ${(+)}_\natty, {(+)}_{\constfont{Bin}}$
and constants $0, 0_{\constfont{Bin}}$,
white-box transporting the term results
in $0_{\constfont{Bin}} +_{\constfont{Bin}} 0_{\constfont{Bin}}$.
These modes can also be mixed:
given just the equivalence
$\tyequiv{\natty}{\binty}{l}{r}$
and correspondence between ${(+)}_{\natty},{(+)}_{\constfont{Bin}}$,
we obtain $(l\app 0) +_{\constfont{Bin}} (l\app 0)$.
Isabelle's Lifting package also supports white-box transports
via the $\transfer$ method~\cite{kuncarthesis}.
While our work is concerned with
black-box transports,
our prototype also contains
experimental support for white-box transports.
This integration will be further polished in future work.

Angiuli et al.~\citet{univalparam3}
establish representation independence results
in Cubical Agda~\cite{cubicalagda}.
Their approach
applies to a restricted variant of
quasi-partial equivalence relations~\cite{quasipers}.
Essentially, they quotient two types by a given correspondence to obtain
a type equivalence between the quotiented types.

Dagand et al.~\citet{depinterop1,depinterop2} introduce what they call
``type-theoretic partial Galois connections'',
which are essentially partial type equivalences
on an enriched
$\alphaty\app \typefont{option}$ type.
They allow for partiality on one
side of the equivalence but not the other.
Their framework is designed for effective
program extraction and implemented
using Coq's typeclass mechanism.

Ringer et al.~\citet{proofrepair}
developed a Coq plugin
to transport proof terms via type equivalences
for inductive types.
Their theory shares similarities
with~\cite{univalparam1,univalparam2},
but it directly transforms proof terms.
This way, one can remove all references to the old datatype
once the proof terms have been transported to the new target type.
This is not readily achievable using other mentioned frameworks,
including ours.

Type equivalences enjoy the property
of having total and mutually inverse transport functions.
This is not the case for partial Galois connections,
which makes the transport of proofs harder.
For example, the parametricity law for equality
$\parenths[\big]{\funrel{T}{\funrel{T}{(\holiff)}}}\app (=)\app (=)$
holds only if $T$ is left-unique and injective.
This is the case if $T$ is described by a type equivalence
but not in general by a Galois connection.
Kun{\v{c}}ar~\citet{kuncarthesis} provides parametricity rules
for all prominent logical connectives.
These rules also apply to our setting and
will be crucial when we polish the integration
of white-box transports in our prototype.

The works mentioned above
all transport terms via certain notions of equivalences.
But there are also other approaches,
particularly in the field of data refinement.
An example is the CoqEAL framework~\cite{coqeal},
which automatically derives parametricity results using typeclass search.
Another one is Isabelle's Autoref framework~\cite{lammichrefine},
which derives relational parametricity results using white-box transports.
The core inspiration in both cases goes back to
\cite{representation_independence,reynoldsparampoly,wadlertheoremsfree}.
A comprehensive comparison of these frameworks
can be found in~\citet{refinementcompare}.

\paragraph{Galois Connections in Computer Science}
Galois connections are fundamental
in the field of abstract interpretation.
Cousot and Cousot's recent book~\citet{abstractintbook}
provides an overview of their applications.
The closure of Galois connections under non-dependent function relators
goes back to at least~\cite{abstractintfirstgal}.
We generalised this result
to partial Galois connections and dependent function relators
in \cref{sec:closuredepfunrel}.
Most work in abstract interpretation
does not consider partially defined Galois connections
and assumes partial orderings on relations.
The work of Min{\'e}~\citet{partgalconn}
is an exception,
allowing for partiality on one side of the connection but not the other.
Darais and Van Horn~\citet{galoisform1} formalise Galois connections
constructively and apply it to tasks in abstract interpretation.
An early
application of Galois connections
was by Hartmanis and Stearns~\citet{pairalgebra2}.
Though they did not use Galois connections,
they introduced an equivalent notion of \emph{pair algebras}~\cite{pairalgebraequivgalcon}.
Our Galois relator
indeed
describes the pair algebra induced by a Galois connection.

\section{Conclusion and Future Work}\label{sec:conclusion}

We explored existing notions of equivalences used for automatic transport.
Based on this exploration,
we identified a set of minimal expectations when transporting terms via equivalences.
This essence led us to introduce a new class of equivalences,
namely partial Galois connections.
Partial Galois connections generalise (standard) Galois connections
and apply to relations
that are only defined on subsets of their types.
We derived closure conditions for partial Galois connections and equivalences,
and typical order properties
under (dependent) function relators, relators for (co)datatypes,
and composition.
Our framework applies to simple type theory and --
unlike prior solutions for simple type theory --
can handle inter-argument dependencies.
We implemented a prototype in Isabelle/HOL based on our results.
The prototype needs to be further polished,
but it can already handle relevant examples that
are out of scope for existing tools.

\paragraph{Future work}
As our theory subsumes the one of Isabelle's Lifting package,
one goal is to replace the package by a more general tool.
To this end,
we have to integrate our results into Isabelle's (co)datatypes package~\cite{bnfimpl},
extend our prototype to automate the construction of compositions,
and polish the support of white-box transports (cf. \cref{sec:relatedwork}).

Finally,
based on our formalisation insights,
we conjecture that one can adopt our theory
to constructive logics,
but only a formalisation in a constructive prover
will give a definite answer.



\anontext{}{
\subsubsection{Acknowledgements}
The
\anontext{authors thank}{author thanks}
the anonymous reviewers of this
and a previous submission
for their valuable feedback
and Mohammad Abdulaziz and Tobias Nipkow
for their comments on a draft of this paper.}

\bibliographystyle{splncs04}
\bibliography{bibliography}

\arxivsubmittext{}{\clearpage}
\appendix
\section{Partial Galois Connections, Equivalences, and Relators}\label{sec:appendixgalconequivrel}

\subsection{(Order) Basics}\label{sec:appendixorderbasics}

Given types $\alphaty$ and $\betaty$,
the type of functions from $\alphaty$ to $\betaty$
is written $\alphaty\purefun\betaty$.

The composition of functions is defined as $(f\comp g)\app x\define f\app(g\app x)$.

A \emph{predicate on a type $\alphaty$} is a function of type
$\alphaty \purefun \bool$.

The predicate mapping all inputs to $\True$ is denoted by $\top\define\lambda\app\_.\app\True$.

A \emph{relation on $\alphaty$ and $\betaty$} is a function of type
$\relty{\alphaty}{\betaty}$.

For every relation $R$, we introduce an infix operator $\inflerel{R}\define R$,
that is $x \lerel{R} y \holiff R\app x\app y$.

The \emph{inverse} of a relation is defined as
$\inv{R}\app x\app y\define R\app y\app x$.

The \emph{composition of two relations $R,S$}
is defined as $(R\relcomp S)\app x\app y\define\exists z.\ R\app x\app z\land S\app z\app y$.

A \emph{relation $R$ is finer than another relation $S$}, written $R\le S$,
if $\forall x\app y.\app R\app x\app y\holimplies S\app x\app y$.

The \emph{domain}, \emph{codomain}, and \emph{field} predicates on a relation $R$
are defined as
\begin{align*}
\indom\app R\app x&\define
\exists y.\app R\app x\app y\\
\incodom\app R\app y&\define
\exists x.\app R\app x\app y\\
\infield\app R\app x&\define
\indom\app R\app x\lor \incodom\app R\app x
\end{align*}

A relation $R$ is \emph{right-total} if
$\forall y.\;\exists x.\; R\app x\app y$
and \emph{right-unique} if
$\forall x,y,y'.\; R\app x\app y\land R\app x\app y' \to y = y'$.

Given a predicate $P$ and relation $R$,
we define \emph{reflexivity}, \emph{transitivity},
and \emph{symmetry on $P$ and $R$} as follows:
\begin{align*}
\reflon\app P\app R&\define
\forall x.\app P\app x \holimplies R\app x\app x\\
\transon\app P\app R&\define
\forall x\app y\app z.\app P\app x\app \land P\app y \land P\app z \land R\app x\app y \land R\app y\app z \holimplies R\app x\app z\\
\symmon\app P\app R&\define
\forall x\app y.\app P\app x\app \land P\app y\land R\app x\app y \holimplies R\app y\app x
\end{align*}
\emph{Preorders} and \emph{partial equivalence relations} (PERs)
are then defined in the expected way:
\begin{align*}
\preorderon\app P\app R&\define
\transon\app P\app R\land \reflon\app P\app R\\
\partequivon\app P\app R&\define
\transon\app P\app R\land \symmon\app P\app R
\end{align*}
For all relativised concepts, we introduce their unrelativised analogue:
\begin{align*}
\refl\app R&\define\reflon\app\top\app R\\
\trans\app R&\define\transon\app\top\app R\\
&\shortvdotswithin{\define}
\partequiv\app R&\define\partequivon\app\top\app R
\end{align*}

Given a predicate $P$ and relation $R$,
we say that $f$ is \emph{inflationary (sometimes also called \emph{extensive}) on $P$ and $R$},
written \makebox{$\inflaton\app P\app R\app f$},
if $\forall x.\app P\app x \holimplies x \lerel{R} f\app x$.
Similarly, we say that $f$ is \emph{deflationary on $P$ and $R$},
written as $\deflaton\app P\app R\app f$,
if $\forall x.\app P\app x \holimplies f\app x\lerel{R} x$.
If $f$ is inflationary and deflationary on $P$ and $R$,
it is a \emph{relational equivalence on $P$ and $R$}:
\begin{displaymath}
\relequivon\app P\app R\app f\define
\inflaton\app P\app R\app f\land
\deflaton\app P\app R\app f.
\end{displaymath}

\subsection{Function Relators and Monotonicity}\label{sec:appendixfunrelmono}

The \emph{dependent function relator} is defined as
\begin{displaymath}
\parenths[\big]{\depfunrel{x}{y}{R}{S}}\app f\app g \define
\forall x\app y.\app R\app x\app y\holimplies S\app(f\app x)\app(g\app y),
\end{displaymath}
where $x,y$ may occur freely in $S$.
The \emph{(non-dependent) function relator} is given as a special case:
$\parenths{\funrel{R}{S}} \define \parenths[\big]{\depfunrel{\_}{\_}{R}{S}}$.
A function is \emph{monotone from $R$ to $S$} if it maps $R$-related inputs to $S$-related outputs:
\begin{displaymath}
\parenths[\big]{\depmono{x}{y}{R}{S}}\app f \define \parenths[\big]{\depfunrel{x}{y}{R}{S}}\app f\app f,
\end{displaymath}
where $x,y$ may occur freely in $S$.
The non-dependent variant is given as a special case:
$\parenths{\mono{R}{S}} \define \parenths[\big]{\depmono{\_}{\_}{R}{S}}$.
A \emph{monotone function relator} is like a function relator but additionally requires its members to be monotone:
\begin{align*}
\parenths[\big]{\monodepfunrel{x}{y}{R}{S}}\app f\app g \define &\parenths[\big]{\depfunrel{x}{y}{R}{S}}\app f\app g\\
                                                                &\land \parenths[\big]{\depmono{x}{y}{R}{S}}\app f
                                                                \land \parenths[\big]{\depmono{x}{y}{R}{S}}\app g,
\end{align*}
where $x,y$ may occur freely in $S$.
The non-dependent variant is given as a special case:
$\parenths{\monofunrel{R}{S}} \define \parenths[\big]{\monodepfunrel{\_}{\_}{R}{S}}$.
We define the \emph{relational if conditional} and
the following notation:
\begin{align*}
\relif\app B\app S\app x\app y &\define B \holimplies S\app x\app y,\\
\parenths[\big]{\depfunrelrest{x}{y}{R}{B}{S}}&\define
\parenths[\big]{\depfunrel{x}{y}{R}{\relif\app B\app S}},
\end{align*}
where in the latter two cases, $x,y$ may occur freely in $B,S$.



\subsection{Galois Relator}\label{sec:appendixgalrel}

We define the dual of $\infgalrel{L}$ as
$\flipinfgalrel{R}\define\galrelconst\app \infgerel{R}\app \infgerel{L}\app l$,
that is $x\flipgalrel{R}y\holiff \indom\app \inflerel{L}\app x \land l\app x \lerel{R} y$.
\begin{lemma}\label{lem:galreliffalt}
Assume $\galp[\big]{\inflerel{L}}{\inflerel{R}}{l}{r}$.
Then
$x \galrel{L} y \holiff x\flipgalrel{R}y$.
\end{lemma}

\subsection{Partial Galois Connections and Equivalences}\label{sec:appendixgalconeqiv}

Typically, Galois connections are defined on preorders,
distinguished by the characteristic property
$x \lerel{L}r\app y \holiff l\app x\lerel{R} y$ for all $x,y$.
We break the concept down into smaller pieces
and lift it to a partial setting.
The \emph{(partial) half Galois property on the left} is defined as
\begin{displaymath}
\halfgalpl[\big]{\inflerel{L}}{\inflerel{R}}{l}{r} \define
\forall x\app y.\app x\galrel{L}y \holimplies l\app x \lerel{R} y
\end{displaymath}
and dually, the \emph{(partial) half Galois property on the right} as
\begin{displaymath}
\halfgalpr[\big]{\inflerel{L}}{\inflerel{R}}{l}{r} \define
\forall x\app y.\app x\flipgalrel{R}y  \holimplies x \lerel{L} r\app y,
\end{displaymath}
Both halves combined constitute the \emph{(partial) Galois property}:
\begin{displaymath}
\galp[\big]{\inflerel{L}}{\inflerel{R}}{l}{r} \define
\halfgalpl[\big]{\inflerel{L}}{\inflerel{R}}{l}{r} \land
\halfgalpr[\big]{\inflerel{L}}{\inflerel{R}}{l}{r}.
\end{displaymath}
If $l$ and $r$ are also monotone,
we obtain a \emph{(partial) Galois connection}:
\begin{displaymath}
\galcapp[\big]{\inflerel{L}}{\inflerel{R}}{l}{r} \define
\parenths[\big]{\mono{\inflerel{L}}{\inflerel{R}}}\app l \land
\parenths[\big]{\mono{\inflerel{R}}{\inflerel{L}}}\app r \land
\galp[\big]{\inflerel{L}}{\inflerel{R}}{l}{r}.
\end{displaymath}
Note that we neither require $\inflerel{L},\inflerel{R}$
to be transitive nor reflexive.
An example Galois connection can be found in \cref{fig:galcon}.

By requiring a two-sided Galois connection,
we obtain a \emph{(partial) Galois equivalence}:
\begin{displaymath}
\galequiv[\big]{\inflerel{L}}{\inflerel{R}}{l}{r} \define
\galcapp[\big]{\inflerel{L}}{\inflerel{R}}{l}{r}\land
\galcapp[\big]{\inflerel{R}}{\inflerel{L}}{r}{l}
\end{displaymath}
An example of a Galois equivalence can be found in \cref{fig:galequiv}.
It can be shown that Galois equivalences are,
in many circumstances,
equivalent to the traditional notion of
(partial) order equivalences~(see \cref{sec:appendixorderequivs}).

Since the relations $\inflerel{L},\inflerel{R}$
are preorders or partial equivalence relations
in many practical cases,
we introduce two more definitions for convenience:
\begin{align*}
\preequiv[\big]{\inflerel{L}}{\inflerel{R}}{l}{r}\define&
\galequiv[\big]{\inflerel{L}}{\inflerel{R}}{l}{r}\\
&\land \preorderon\app\infieldapp{\inflerel{L}}\\
&\land \preorderon\app\infieldapp{\inflerel{R}}\\
\perequiv[\big]{\inflerel{L}}{\inflerel{R}}{l}{r}\define&
\galequiv[\big]{\inflerel{L}}{\inflerel{R}}{l}{r}\\
&\land \partequivon\app\infieldapp{\inflerel{L}}\\
&\land \partequivon\app\infieldapp{\inflerel{R}}
\end{align*}

\subsubsection{Order Equivalences}\label{sec:appendixorderequivs}

To define the concept of an order equivalence,
we first define the \emph{unit} and \emph{counit} functions:
\begin{displaymath}
\unitconst\app l\app r \define r\comp l\qquad\qquad\qquad
\counitconst\app l\app r\define l\comp r
\end{displaymath}
When $l$ and $r$ are clear from the context, we will
write $\unit\define\unitconst\app l\app r$
and $\counit\define\counitconst\app l\app r$.
A \emph{(partial) order equivalence} is then defined as
\begin{displaymath}
\begin{aligned}
\orderequiv[\big]{\inflerel{L}}{\inflerel{R}}{l}{r} \define &
\parenths[\big]{\mono{\inflerel{L}}{\inflerel{R}}}\app l
\land \parenths[\big]{\mono{\inflerel{R}}{\inflerel{L}}}\app r\\
&\land \relequivon\app (\infield\app \inflerel{L})\app \inflerel{L}\app \unit\\
&\land \relequivon\app (\infield\app \inflerel{R})\app \inflerel{R}\app \counit.
\end{aligned}
\end{displaymath}
In practice,
we will commonly work with preorders,
where the notions of Galois equivalences and order equivalences coincide:
\begin{lemma}
Assume
\begin{assumes}(3)
\item $\orderequiv[\big]{\inflerel{L}}{\inflerel{R}}{l}{r}$,
\item $\trans\app\inflerel{L}$,
\item $\trans\app\inflerel{R}$.
\end{assumes}
Then $\galequiv[\big]{\inflerel{L}}{\inflerel{R}}{l}{r}$.
\end{lemma}
\begin{lemma}
Assume
\begin{assumes}(2)
\item $\galequiv[\big]{\inflerel{L}}{\inflerel{R}}{l}{r}$,
\item $\reflon\app\infieldapp{\inflerel{L}}$,
\item $\reflon\app\infieldapp{\inflerel{R}}$.
\end{assumes}
Then $\orderequiv[\big]{\inflerel{L}}{\inflerel{R}}{l}{r}$.
\end{lemma}

\section{Closure Properties}\label{sec:appendixclosure}

\subsection{(Dependent) Function Relator}

In \cref{sec:closuredepfunrel},
we only stated our results
for Galois equivalences on preorders and partial equivalence relations
and the dependent function relator.
In this section,
we show the more general results for Galois connections
for both the (non-dependent) and dependent function relator.
We also clarify the need of the monotone function relator.

\subsubsection{Function Relator}\label{sec:appendixclosurefunrel}

In practice, the relations and functions
we use
are often non-dependent.
The definitions in \labelcref{eq:closuredepfunreldef} then simplify to the standard,
non-dependent function relator and mapper.
Moreover, the closure theorems will have
considerably simpler assumptions.
We hence find it instructive to present the
results for this special case.
Let us fix the following variables:
\begin{align*}
\begin{split}
L_1 &\holhasty \alphaty_1 \purefun \alphaty_1 \purefun \bool,\\
R_1 &\holhasty \alphaty_2 \purefun \alphaty_2 \purefun \bool,\\
\nondepcase{L_2} &\holhasty \betaty_1 \purefun \betaty_1 \purefun \bool,\\
\nondepcase{R_2} &\holhasty \betaty_2 \purefun \betaty_2 \purefun \bool,\\
\end{split}
\begin{split}
l_1 &\holhasty \alphaty_1 \purefun \alphaty_2,\\
r_1 &\holhasty \alphaty_2 \purefun \alphaty_1,\\
\nondepcase{l_2} &\holhasty \betaty_1 \purefun \betaty_2,\\
\nondepcase{r_2} &\holhasty \betaty_2 \purefun \betaty_1.
\end{split}
\end{align*}
Compared to \cref{eq:closuredepfunreldef},
the target relations and transport functions then simplify to
\begin{align*}
\begin{split}
\nondepcase{L} &\define
\parenths[\big]{\monofunrel{\inflerel{L_1}}{\inflerel{\nondepcase{L_2}}}},\\
\nondepcase{l} &\define
\parenths[\big]{\funmap{r_1}{\nondepcase{l_2}}},
\end{split}
\begin{split}
\nondepcase{R} &\define
\parenths[\big]{\monofunrel{\inflerel{R_1}}{\inflerel{\nondepcase{R_2}}}},\\
\nondepcase{r} &\define
\parenths[\big]{\funmap{l_1}{\nondepcase{r_2}}},
\end{split}
\end{align*}
where
$\parenths{\funmap{f}{g}} \define \parenths{\depfunmap{\_}{f}{g}}$
is the \emph{(non-dependent) function mapper}.
In other words: $\parenths{\funmap{f}{g}}\app h = g\comp h\comp f$.

\begin{lemma}\label{thm:funrel_galc}
Assume
\begin{assumes}(2)
\item\label{asm:funrel_galcon1}$\galcapp{\inflerel{L_1}}{\inflerel{R_1}}{l_1}{r_1}$,
\item\label{asm:funrel_refl1}$\reflon\app\infieldapp{\inflerel{L_1}}$,
\item $\reflon\app\infieldapp{\inflerel{R_1}}$,
\item\label{asm:funrel_galcon2}$\galcapp{\inflerel{\nondepcase{L_2}}}{\inflerel{\nondepcase{R_2}}}{\nondepcase{l_2}}{\nondepcase{r_2}}$,
\item $\trans\app\inflerel{\nondepcase{L_2}}$,
\item\label{asm:funrel_transr2}$\trans\app\inflerel{\nondepcase{R_2}}$.
\end{assumes}
Then $\galcapp[\big]{\inflerel{\nondepcase{L}}}{\inflerel{\nondepcase{R}}}{\nondepcase{l}}{\nondepcase{r}}$.
\end{lemma}
\begin{proof}
The theorem is a direct consequence of
\cref{thm:depfunrel_galc},
but it is instructive to consider
the proof of this simpler theorem first.
We only show the case for
$\halfgalpl[\big]{\inflerel{\nondepcase{L}}}{\inflerel{\nondepcase{R}}}{\nondepcase{l}}{\nondepcase{r}}$.
The other cases are similar.
Assume
\begin{assumes}[label=\alph*,ref=\alph*](3)
\item\label{asm:funrel_nondepfunincodom}$\incodom\app \inflerel{\nondepcase{R}}\app g$,
\item\label{asm:funrel_nondepfunle}$f \lerel{\nondepcase{L}} \nondepcase{r}\app g$,
\item\label{asm:funrel_nondepfunr1}$x_1'\lerel{R_1} x_2'$.
\end{assumes}
Our goal is
$\nondepcase{l}\app f\app x_1' \lerel{\nondepcase{R_2}} g\app x_2'$.
Due to monotonicity of $r_1$ (\cref{asm:funrel_galcon1}),
we get
\makebox{$r_1\app x_1'\lerel{L_1} r_1\app x_2'$}.
Due to \cref{asm:funrel_nondepfunle},
we get
\begin{displaymath}
f\app (r_1\app x_1')\lerel{\nondepcase{L_2}} \nondepcase{r}\app g\app (r_1\app x_2')=\nondepcase{r_2}\app \parenths[\big]{g\app (\counit_1\app x_2')}.
\end{displaymath}
Since
$\halfgalpl[\big]{\inflerel{\nondepcase{L_2}}}{\inflerel{\nondepcase{R_2}}}{\nondepcase{l_2}}{\nondepcase{r_2}}$
(\cref{asm:funrel_galcon2}),
we get
\begin{displaymath}
\nondepcase{l_2}\app \parenths[\big]{f\app (r_1\app x_1')}
= \nondepcase{l}\app f\app x_1'\lerel{\nondepcase{R_2}} g\app (\counit_1\app x_2').\footnote{%
Strictly speaking, we first need to show that
$\incodom\app\inflerel{R_2}\app\parenths[\big]{g\app(\counit_1\app x_2')}$,
but we omit that technical step here.}
\end{displaymath}
Due to transitivity (\cref{asm:funrel_transr2}),
it remains to show that
$g\app (\counit_1\app x_2')\lerel{\nondepcase{R_2}} g\app x_2'$.
This follows from the first Galois connection,
reflexivity of $\inflerel{R_1}$,
monotonicity of $g$,
and $\incodom\app\inflerel{R_1}\app x_2'$
(\crefns{asm:funrel_galcon1,asm:funrel_refl1,asm:funrel_nondepfunincodom,asm:funrel_nondepfunr1}).
\end{proof}
Specialising \cref{thm:depfunrel_galc_galreleq} to the non-dependent function relator yields:
\begin{lemma}
Assume
\begin{assumes}(2)
\item$\galcapp{\inflerel{L_1}}{\inflerel{R_1}}{l_1}{r_1}$,
\item $\reflon\app\infieldapp{\inflerel{L_1}}$,
\item $\parenths[\big]{\mono{\inflerel{\nondepcase{R_2}}}{\inflerel{\nondepcase{L_2}}}}\app \nondepcase{r_2}$,
\item$\trans\app\inflerel{\nondepcase{L_2}}$,
\item $\indom\app \inflerel{\nondepcase{L}}\app f$,
\item $\incodom\app \inflerel{\nondepcase{R}}\app g$.
\end{assumes}
Then
$f \galrel{\nondepcase{L}} g\holiff
\parenths[\big]{\funrel{\infgalrel{L_1}}{\infgalrel{\nondepcase{L_2}}}}f\app g$.
\end{lemma}

\subsubsection{Dependent Function Relator}\label{sec:appendixclosuredepfunrel}

As in \cref{thm:depfunrel_galequiv},
the closure theorem
requires monotonicity conditions for each of the dependent variables
(\cref{asm:depfunrel_monoleft2,asm:depfunrel_monoright2,asm:depfunrel_monol2,asm:depfunrel_monor2} below).
Morally speaking, these assumptions say that
\begin{enumerate*}
\item $L_2$ is antimonotone in its first and restricted antimonotone in its second parameter,
\item $R_2$ is restricted monotone in its first and monotone in its second parameter, and
\item $l_2, r_2$ are monotone in both parameters.
\end{enumerate*}
\begin{theorem}\label{thm:depfunrel_galc}
Define $\unit_1\define \unitconst\app l_1\app r_1$ and
$\counit_1\define \counitconst\app l_1\app r_1$.
Assume
\begin{assumes}(1)
\item\label{asm:depfunrel_galcon1}$\galcapp{\inflerel{L_1}}{\inflerel{R_1}}{l_1}{r_1}$,
\item $\reflon\app\infieldapp{\inflerel{L_1}}$,
\item\label{asm:depfunrel_relfr1}$\reflon\app\infieldapp{\inflerel{R_1}}$,
\item\label{asm:depfunrel_galcon2}if $x \galrel{L_1} x'$ then $\galcapp{\inflerel{L_2\app x\app (r_1\app x')}}{\inflerel{R_2\app (l_1\app x)\app x'}}{(\depfunleft{x'}{x})}{(\depfunright{x}{x'})}$,
\item\label{asm:depfunrel_transl2}if $x_1 \lerel{L_1} x_2$ then $\trans\app\inflerel{L_2\app x_1\app x_2}$,
\item\label{asm:depfunrel_transr2}if $x_1' \lerel{R_1} x_2'$ then $\trans\app\inflerel{R_2\app x_1'\app x_2'}$,
\item\label{asm:depfunrel_monoleft2}%
if $x_1 \lerel{L_1} x_2 \lerel{L_1} x_3 \lerel{L_1} x_4\lerel{L_1} \unit_1\app x_3$ then
$\inflerel{L_2\app x_2\app x_4} \leq \inflerel{L_2\app x_1\app x_3}$,
\item\label{asm:depfunrel_monoright2}%
if $\counit_1\app x_2'\lerel{R_1} x_1' \lerel{R_1} x_2' \lerel{R_1} x_3' \lerel{R_1} x_4'$ then
$\inflerel{R_2\app x_1'\app x_3'} \leq \inflerel{R_2\app x_2'\app x_4'}$,
\item\label{asm:depfunrel_monol2}%
if $x_1 \lerel{L_1} x_2 \galrel{L_1} x_1' \lerel{R_1} x_2'$ and $\infield\app\inflerel{L_2\app x_1\app (r_1\app x_2')}\app y$ then\\
$\parenths[\big]{\depfunleft{x_1'}{x_1}\app y} \lerel{R_2\app (l_1\app x_1)\app x_2'} \parenths[\big]{\depfunleft{x_2'}{x_2}\app y}$,
\item\label{asm:depfunrel_monor2}%
if $x_1 \lerel{L_1} x_2 \galrel{L_1} x_1' \lerel{R_1} x_2'$ and $\infield\app\inflerel{R_2\app (l_1\app x_1)\app x_2'}\app y'$ then\\
$\parenths[\big]{\depfunright{x_1}{x_1'}\app y'} \lerel{L_2\app x_1\app (r_1\app x_2')} \parenths[\big]{\depfunright{x_2}{x_2'}\app y'}$.
\end{assumes}
Then $\galcapp[\big]{\inflerel{L}}{\inflerel{R}}{l}{r}$.
\end{theorem}
\begin{proof}
We will only prove that
$\halfgalpl{\inflerel{L}}{\inflerel{R}}{l}{r}$.
This should primarily illustrate how the monotonicity
requirements arise as part of the proof.
The rest of the proof can be found in our formalisation.
It is also instructive to first consider the proof for the
non-dependent function relator as
it uses the same core ideas (see \cref{thm:funrel_galc}).

\begin{figure}[t]
\centering
\resizebox{0.90\textwidth}{!}{%
  \begin{tikzpicture}
  \pic (alpha2) {righttype={$\inflerel{R_1}$}};

  \node[element] (x1') at ($(alpha2type) + (0, 0.5)$) {};
  \node[above=\labeldistance of x1'] {$x_1'$};
  \node[element] (x2') at ($(x1') + (0, -\equivdistance)$) {};
  \node[below=\labeldistance of x2'] {$x_2'$};

  \pic (alpha1) at ($(alpha2type) + (0, \defaulttypedistancey)$) {lefttype={{$\inflerel{L_1}$}}};

  \node[element] (r1x1') at ($(x1') + (0,\defaulttypedistancey)$) {};
  \node[left=\labeldistance of r1x1',yshift=0.1cm] {$r_1\app x_1'$};

  \node[type] (beta1type) at ($(alpha1type) + (1.2*\defaulttypedistance, 0)$) {};
  \pic (beta1label) [above=\labeldistance of beta1type,xshift=-1.3cm] {labelrel={leftlabelstyle2}{$\inflerel{L_2\app(r_1\app x_1')\app(r_1\app x_1')}$}};

  \node[element] (fr1x1') at ($(r1x1') + (1.2*\defaulttypedistance,0.5*\equivdistance)$) {};
  \node[above=\labeldistance of fr1x1'] {$f\app (r_1\app x_1')$};
  \node[element] (rgr1x1') at ($(fr1x1') + (0,-\equivdistance)$) {};
  \node[right=\labeldistance of rgr1x1'] {$\depfunright{(r_1\app x_1')}{(\counit_1\app x_1')}\app \parenths{g\app(\counit_1\app x_1')}$};
  \node[element] (rgr1x1'adapt) at ($(rgr1x1') + (0,-\equivdistance)$) {};
  \node[right=\labeldistance of rgr1x1'adapt] {$\depfunright{(r_1\app x_1')}{x_1'}\app \parenths{g\app(\counit_1\app x_1')}$};

  \node[type] (beta2type) at ($(beta1type) + (0,-\defaulttypedistancey)$) {};
  \pic (beta2label) [above=\labeldistance of beta2type,xshift=-0.9cm] {labelrel={rightlabelstyle2}{$\inflerel{R_2\app(\counit_1\app x_1')\app x_1'}$}};

  \node[element] (lfx1') at ($(fr1x1') + (0,-\defaulttypedistancey,0)$) {};
  \node[above=\labeldistance of lfx1'] {$l\app f\app x_1'$};
  \node[element] (gcounitx1') at ($(lfx1') + (0,-\equivdistance,0)$) {};
  \node[right=\labeldistance of gcounitx1'] {$g(\counit_1\app x_1')$};
  \node[element] (gx1') at ($(gcounitx1') + (0,-\equivdistance,0)$) {};
  \node[below=\labeldistance of gx1'] {$g\app x_1'$};

  \node[type] (beta2type') at ($(beta2type) + (1.2*\defaulttypedistance,0)$) {};
  \pic (beta2label') [above=\labeldistance of beta2type',xshift=-0.9cm] {labelrel={rightlabelstyle2}{$\inflerel{R_2\app x_1'\app x_2'}$}};

  \node[element] (lfx1'adapt) at ($(lfx1') + (1.2*\defaulttypedistance,0)$) {};
  \node[above=\labeldistance of lfx1'adapt] {$l\app f\app x_1'$};
  \node[element] (gx1'adapt) at ($(gx1') + (1.2*\defaulttypedistance,0.5*\equivdistance)$) {};
  \node[right=\labeldistance of gx1'adapt] {$g\app x_1'$};
  \node[element] (gx2') at ($(gx1'adapt) + (0,-\equivdistance)$) {};
  \node[below=\labeldistance of gx2'] {$g\app x_2'$};

  \draw[->,relstyle,rightstyle] (x1') to (x2');
  \draw[->,loop right,loopstyle,rightstyle,relstyle] (x1') to (x1');
  \draw[->,loop above,loopstyle,leftstyle,relstyle] (r1x1') to (r1x1');
  \draw[->] (x1') edge[rightstyle,edgestyle,bend left=55,"$r_1$"] (r1x1');
  \draw[->] (r1x1') edge[asmedge,"$f$",pos=0.55] (fr1x1');
  \draw[->] (r1x1') edge[asmedge,"$r\app g$" swap,pos=0.60] (rgr1x1');
\begin{pgfinterruptboundingbox}
  \draw[->] (x1') edge[asmedge,bend right=75,"$g$",pos=0.3] (gx1'adapt);
  \draw[->] (x2') edge[asmedge,bend right=42] (gx2');
\end{pgfinterruptboundingbox}
  \node[none] (boundingboxbottom) at ($(beta2type.south) + (0, -\typeequivlabeldistance)$) {};
  \draw[->,relstyle2,leftstyle2] (fr1x1') to (rgr1x1');
  \draw[->,relstyle2,leftstyle2] (rgr1x1') to (rgr1x1'adapt);
  \draw[->,relstyle2,rightstyle2] (lfx1') to (gcounitx1');
  \draw[->,relstyle2,rightstyle2] (gcounitx1') to (gx1');
  \draw[->,relstyle2,rightstyle2] (lfx1'adapt) to (gx1'adapt);
  \draw[->,relstyle2,rightstyle2] (gx1'adapt) to (gx2');
  \draw[->] (lfx1') edge[asmedge,"$Asm.\ \labelcref{asm:depfunrel_monoright2}$"] (lfx1'adapt);
  \draw[->] (gx1') edge[asmedge,"$Asm.\ \labelcref{asm:depfunrel_monoright2}$" swap,pos=0.32] (gx1'adapt);
  \draw[->] (fr1x1') edge[asmedge,bend left=-65,"$Asm.\ \labelcref{asm:depfunrel_galcon2}$" swap,pos=0.60] (lfx1');
  \draw[->] (rgr1x1'adapt) edge[asmedge,bend left=-80] (gcounitx1');
\end{tikzpicture}}
\caption{Proof of $\halfgalpl{\inflerel{L}}{\inflerel{R}}{l}{r}$ as explained in \cref{thm:depfunrel_galc}.
Types are drawn solid, black,
transport functions dashed,
relations dotted and dashed-dotted.
}\label{fig:depfunrel_galc}
\Description[Proof of $\halfgalpl{\inflerel{L}}{\inflerel{R}}{l}{r}$ as explained in \cref{thm:depfunrel_galc}.]{$\halfgalpl{\inflerel{L}}{\inflerel{R}}{l}{r}$ as explained in \cref{thm:depfunrel_galc}. The proof goes via all involved relations and uses the monotonicity assumptions for $r_2$ and $R_2$.}
\end{figure}
A visualisation of the following proof can be found in
\cref{fig:depfunrel_galc}.
Assume
\begin{assumes}[label=\alph*,ref=\alph*](3)
\item\label{asm:depfunrel_depfunincodom}$\incodom\app\inflerel{R}\app g$,
\item\label{asm:depfunrel_depfunle}$f \lerel{L} r\app g$,
\item\label{asm:depfunrel_depfunr1}$x_1'\lerel{R_1} x_2'$.
\end{assumes}
We have to show that
$\parenths[\big]{l\app f\app x_1'} \lerel{R_2\app x_1'\app x_2'} \parenths[\big]{g\app x_2'}$,
which unfolds to
\begin{displaymath}
\parenths[\big]{\depfunleft{x_1'}{(r_1\app x_1')}\app \parenths[\big]{f\app (r_1\app x_1')}}
\lerel{R_2\app x_1'\app x_2'} \parenths[\big]{g\app x_2'}.
\end{displaymath}
First we apply reflexivity of $\inflerel{R_1}$
to obtain $x_1'\lerel{R_1} x_1'$.
With monotonicity of $r_1$ (\cref{asm:depfunrel_galcon1}),
we get
$r_1\app x_1'\lerel{L_1} r_1\app x_1'$.
Due to \cref{asm:depfunrel_depfunle},
we get
\begin{displaymath}
\parenths[\big]{f\app (r_1\app x_1')}\lerel{L_2\app(r_1\app x_1')\app(r_1\app x_1')} \parenths[\big]{r\app g\app (r_1\app x_1')}=\depfunright{(r_1\app x_1')}{(\counit_1\app x_1')}\app \parenths[\big]{g\app (\counit_1\app x_1')}.
\end{displaymath}
Now unlike in \cref{thm:funrel_galc},
we cannot directly apply \cref{asm:depfunrel_galcon2}:
the parameters of $\inflerel{L_2\app(r_1\app x_1')\app(r_1\app x_1')}$ and $\depfunright{(r_1\app x_1')}{(\counit_1\app x_1')}$ do not match up.
We first have to use monotonicity of $r_2$ (\cref{asm:depfunrel_monor2}) to obtain
\begin{displaymath}
\parenths[\big]{\depfunright{(r_1\app x_1')}{(\counit_1\app x_1')}\app \parenths[\big]{g\app (\counit_1\app x_1')}}\lerel{L_2\app(r_1\app x_1')\app(r_1\app x_1')}
\parenths[\big]{\depfunright{(r_1\app x_1')}{x_1'}\app \parenths[\big]{g\app (\counit_1\app x_1')}}.
\end{displaymath}
With transitivity (\cref{asm:depfunrel_transl2}),
we then get
\begin{displaymath}
\parenths[\big]{f\app (r_1\app x_1')}\lerel{L_2\app(r_1\app x_1')\app(r_1\app x_1')}
\parenths[\big]{\depfunright{(r_1\app x_1')}{x_1'}\app \parenths[\big]{g\app (\counit_1\app x_1')}}.
\end{displaymath}
Now we apply \cref{asm:depfunrel_galcon2} to obtain
\begin{displaymath}
\depfunleft{x_1'}{(r_1\app x_1')}\app \parenths[\big]{f\app (r_1\app x_1')}=(l\app f\app x_1')
\lerel{R_2\app (\counit_1\app x_1')\app x_1'}
\parenths[\big]{g\app (\counit_1\app x_1')}.\footnote{%
Again, we omit the step showing that
$\incodom\app\inflerel{R_2\app (\counit_1\app x_1')\app x_1'}\app\parenths[\big]{g\app(\counit_1\app x_2')}$.}
\end{displaymath}
With monotonicity of $g$ and \cref{asm:depfunrel_galcon1}, one can show that
\begin{displaymath}
\parenths[\big]{g\app (\counit_1\app x_1')}
\lerel{R_2\app (\counit_1\app x_1')\app x_1'} \parenths{g\app x_1'}.
\end{displaymath}
Thus with transitivity (\cref{asm:depfunrel_transr2}),
$(l\app f\app x_1') \lerel{R_2\app (\counit_1\app x_1')\app x_1'} (g\app x_1')$.
Using monotonicity of $R_2$ (\cref{asm:depfunrel_monoright2}),
we can adapt the parameters of $R_2$ and obtain
$(l\app f\app x_1') \lerel{R_2\app x_1'\app x_2'} (g\app x_1')$.
Finally, we obtain $(g\app x_1') \lerel{R_2\app x_1'\app x_2'} (g\app x_2')$ from $x_1'\lerel{R_1} x_2'$ and monotonicity of $g$.
We can conclude using transitivity.
\end{proof}

We can also prove a generalisation of
\cref{thm:depfunrel_galreleq}:
\begin{theorem}\label{thm:depfunrel_galc_galreleq}
Assume
\begin{assumes}
\item$\galcapp{\inflerel{L_1}}{\inflerel{R_1}}{l_1}{r_1}$,
\item $\reflon\app\infieldapp{\inflerel{L_1}}$,
\item if $x \galrel{L_1} x'$ then $\parenths[\big]{\mono{\inflerel{R_2\app (l_1\app x)\app x'}}{\inflerel{L_2\app x\app (r_1\app x')}}}\app (\depfunright{x}{x'})$,
\item if $x_1 \lerel{L_1} x_2$ then $\trans\app\inflerel{L_2\app x_1\app x_2}$,
\item%
if $x_1 \lerel{L_1} x_2 \lerel{L_1} x_3$ then
$\inflerel{L_2\app x_1\app x_2} \leq \inflerel{L_2\app x_1\app x_3}$,
\item%
if $x_1\lerel{L_1}x_2 \lerel{L_1} x_3\lerel{L_1}\unit_1\app x_2$ then
$\inflerel{L_2\app x_1\app x_3} \leq \inflerel{L_2\app x_1\app x_2}$,
\item%
if $x_1 \lerel{L_1} x_2 \galrel{L_1} x_1' \lerel{R_1} x_2'$ and $\infield\app\inflerel{R_2\app (l_1\app x_1)\app x_2'}\app y'$ then\\
$\parenths[\big]{\depfunright{x_1}{x_1'}\app y'} \lerel{L_2\app x_1\app (r_1\app x_2')} \parenths[\big]{\depfunright{x_2}{x_2'}\app y'}$,
\item $\indom\app \inflerel{L}\app f$, and $\incodom\app \inflerel{R}\app g$.
\end{assumes}
Then
$f \galrel{L} g\holiff
\parenths[\big]{\depfunrel{x}{x'}{\infgalrel{L_1}}{\infgalrel{L_2\app x\app x'}}}f\app g$.
\end{theorem}

\subsubsection{Regarding Monotonicity}\label{sec:appendixclosuredepfunrelmononote}

Finally,
we want to mention a subtlety:
while work in abstract interpretation points out the necessity to use
\emph{monotone} function relators, for example \cite{abstractint1},
related work dealing with the concept of transports in proof assistants
does not talk about any such monotonicity restriction~\cite{isabellelifting,univalparam1,univalparam2,univalparam3,proofrepair,depinterop1,depinterop2}.
The reason is not that the monotonicity restriction is unnecessary,
but rather that the function relators in latter works are monotone by default.
This can be made precise with the following lemma:
\begin{lemma}\label{lem:depfunmonoeqifpartequiv}
Assume
\begin{assumes}(1)
\item $\reflon\app\infieldapp{\inflerel{L_1}}$,
\item if $x_1 \lerel{L_1} x_2$ then $\inflerel{L_2\app x_2\app x_2} \le \inflerel{L_2\app x_1\app x_2}$,
\item if $x_1 \lerel{L_1} x_2$ then $\inflerel{L_2\app x_1\app x_1} \le \inflerel{L_2\app x_1\app x_2}$,
\item if $x_1 \lerel{L_1} x_2$ then $\partequiv\app\inflerel{L_2\app x_1\app x_2}$.
\end{assumes}
Then
$\parenths[\big]{\monodepfunrel{x_1}{x_2}{\inflerel{L_1}}{\inflerel{L_2\app x_1\app x_2}}}=
\parenths[\big]{\depfunrel{x_1}{x_2}{\inflerel{L_1}}{\inflerel{L_2\app x_1\app x_2}}}$.
\end{lemma}
Again, we can specialise this to the non-dependent function relator:
\begin{lemma}
Assume
\begin{assumes}(1)
\item $\reflon\app\infieldapp{\inflerel{L_1}}$,
\item $\partequiv\app\inflerel{L_2}$.
\end{assumes}
Then
$\parenths[\big]{\monofunrel{\inflerel{L_1}}{\inflerel{\nondepcase{L_2}}}}=
\parenths[\big]{\funrel{\inflerel{L_1}}{\inflerel{\nondepcase{L_2}}}}$.
\end{lemma}
It is easy to check that these assumptions are met
by type equivalences and partial quotient types.

\subsection{Compositions}\label{sec:appendixclosurecomp}

In this section,
we provide some intuition for the
constructions from \cref{sec:closurecomp},
provide preciser results,
and compare the construction
with Isabelle's Lifting package.

\subsubsection{Closure for Coinciding Relations}\label{sec:appendixclosurecompcoincide}
\begin{theorem}\label{thm:comp_galc_coincide}
Let $\star \in\{\galcsym,\galequivsym,\orderequivsym,\preequivsym,\perequivsym\}$ and assume
\begin{assumes}(3)
\item $\parenths[\big]{\inflerel{L_1}\star\inflerel{R_1}}\app l_1\app r_1$,
\item $\parenths[\big]{\inflerel{R_1}\star\inflerel{R_2}}\app l_2\app r_2$,
\item $\inflerel{R_1}=\inflerel{L_2}$.
\end{assumes}
Then \makebox{$\parenths[\big]{\inflerel{L_1}\star\inflerel{R_2}}\app (l_2\comp l_1)\app (r_1\comp r_2)$}.
\end{theorem}
\begin{proof}
The proof can be found in the formalisation\footnote{
We actually prove a more general result
where the right and left relations of the input Galois connections
need not be equal but only need to ``agree whenever required''.
But we suspect that such an agreement rarely holds in practice
and hence omit it.}.
\end{proof}

\subsubsection{Construction Idea}
As mentioned in \cref{sec:closurecomp},
our construction is inspired by Huffman and and Kun{\v{c}}ar's construction in~\citet{isabellelifting}.
Unfortunately, they
do not provide any intuition about their constructions,
nor does Kun{\v{c}}ar~\citet{kuncarthesis} in his thesis.
We try our best to fill this gap:
In the following,
we call $\inflerel{L_1}$ the \emph{leftmost relation},
$\inflerel{R_1},\inflerel{L_2}$ the \emph{middle relations}, and
$\inflerel{R_2}$ the \emph{rightmost relation}.
We will explain the definition of $\inflerel{L}$.
The case for $\inflerel{R}$ is symmetric.

Fix some $x\holhasty\alphaty$ of the leftmost type.
We want to
\begin{enumerate*}[label=(\alph*)]
\item\label{propy:compapplyldom}make sure that applying $l=l_2\circ l_1$ on $x$
does not leave the domain/codomain of our equivalences, and
\item find all elements $x'\holhasty\alphaty$ that are greater
or equal than $x$ while doing so.
\end{enumerate*}
We make a first approximation to satisfy these conditions using three ``chase'' steps:
\begin{step}
\item check whether $\indom\app\inflerel{L_1}\app x$ and find some $y$ such that $l_1\app x\lerel{R_1} y$,
\item find some $y'$ such that $y\lerel{L_2}y'$, and
\item\label{step:righttoleft}check whether $\indom\app\inflerel{R_1}\app y'$
and find some $x'$ such that $r_1\app y'\lerel{L_1} x'$.
\end{step}
These steps are not enough:
we may have $l_1\app x\lerel{R_1} y\lerel{L_2}\app y'$
but not necessarily $l_1\app x\lerel{L_2} y\lerel{L_2}\app y'$,
as required for Property~\ref{propy:compapplyldom} and \cref{step:righttoleft}.
But if we further require that $\inflerel{R_1}$ and $\inflerel{L_2}$ \emph{commute},
that is $\parenths[\big]{\inflerel{R_1} \relcomp \inflerel{L_2}} = \parenths[\big]{\inflerel{L_2} \relcomp \inflerel{R_1}}$,
the steps become sufficient.
Finally note that
\begin{itemize}
\item $x \galrel{L_1} y \holiff \indom\app \inflerel{L_1}\app x \land l_1\app x \lerel{R_1} y$ whenever $\galp{\inflerel{L_1}}{\inflerel{R_1}}{l_1}{r_1}$, and
\item $y' \galrel{R_1} x' \holiff \indom\app \inflerel{R_1}\app y' \land r_1\app y' \lerel{L_1} x'$
whenever $\galp{\inflerel{R_1}}{\inflerel{L_1}}{r_1}{l_1}$
\end{itemize}
due to \cref{lem:galreliffalt}.
For Galois equivalences $\galequiv[\big]{\inflerel{L_1}}{\inflerel{R_1}}{l_1}{r_1}$,
it is thus sufficient to search for a chain
$x\galrel{L_1}y\lerel{L_2}y'\galrel{R_1}x'$,
which is equivalent to
$\parenths[\big]{\infgalrel{L_1} \relcomp \inflerel{L_2} \relcomp \infgalrel{R_1}}\app x\app x'$.
Hence the definition of $\inflerel{L}$.
\begin{remark}
A Galois connection
\makebox{$\galcapp[\big]{\inflerel{L_1}}{\inflerel{R_1}}{l_1}{r_1}$}
would not be sufficient due to \cref{step:righttoleft}:
We are given some $y'\holhasty\betaty$  and $x'\holhasty\alphaty$
and need to check whether $y'$ is ``smaller'' than $x'$.
We may check this by either transporting $y'$ to the left
(i.e.\ $r_1\app y'\lerel{L_1} x'$)
or $x'$ to the right (i.e.\ $y'\lerel{R_1}l_1\app x'$).
However, right adjoints only preserve infima
while left adjoints only preserve suprema.
Hence the need for
\makebox{$\galcapp[\big]{\inflerel{R_1}}{\inflerel{L_1}}{r_1}{l_1}$}.

Now it is not to be excluded that there is an alternative way
that avoids the need of a Galois equivalence.
But at least thus far,
it has eluded the \anontext{authors}{author}.
\end{remark}
\begin{remark}
As noted, the relations
$\inflerel{L}$ and $\inflerel{R}$
may not be equal to
$\inflerel{L_1}$ and $\inflerel{R_2}$,
but, in some sense, describe those parts
that were made ``compatible'' with respect to $l$ and $r$.
While our formalisation includes conditions
under which we can obtain an equality,
they do not apply to all practical examples.
It is indeed a challenge on its own to find
particular conditions under which the relations
$\inflerel{L}$ and $\inflerel{R}$
may be rewritten to a simpler form.
In this direction, the
thesis of Kun{\v{c}}ar~\citet{kuncarthesis}
includes ideas applicable to total quotients and partial subtypes.
\end{remark}

\subsubsection{Closure and Similarity Theorems}
The next result generalises \cref{thm:comp_galequiv}.
\begin{theorem}\label{thm:comp_galequiv_better}
Assume
\begin{assumes}(2)
\item\label{asm:comp_galequiv}$\galequiv[\big]{\inflerel{L_i}}{\inflerel{R_i}}{l_i}{r_i}$  for $i\in\{1,2\}$,
\item $\preorderon\app\infieldapp{\inflerel{R_1}}$,
\item\label{asm:comp_preorderleft}$\preorderon\app\infieldapp{\inflerel{L_2}}$,
\item\label{asm:comp_middlecompatcond}$\parenths[\big]{\inflerel{R_1} \relcomp \inflerel{L_2}} = \parenths[\big]{\inflerel{L_2} \relcomp \inflerel{R_1}}$.
\end{assumes}
Then
$\galcapp[\big]{\inflerel{L}}{\inflerel{R}}{l}{r}$.
\end{theorem}
\begin{proof}
We will only show that
$\halfgalpl{\inflerel{L}}{\inflerel{R}}{l}{r}$
to illustrate the usage of the compatibility condition
(\cref{asm:comp_middlecompatcond}).
The rest of the proof can be found in our formalisation.
A visualisation of the following proof can be found in
\cref{fig:comp_galequiv}.

\begin{figure}[!t]
\begin{subfigure}[t]{0.49\textwidth}
  \centering
\resizebox{\textwidth}{!}{%
  \begin{tikzpicture} 
  \node[typesmall] (alphatype) {};
  \pic (alphalabel) [above=\labeldistance of alphatype,xshift=-0.6cm] {labelrel={leftstyle}{$\inflerel{L_1}$}};

  \node[element] (x) at ($(alphatype) + (0, 1.4)$) {};
  \node[above=\labeldistance of x] {$x$};
  \node[element] (rz) at ($(x) + (0, -3*\equivdistance)$) {};
  \node[below=\labeldistance of rz] {$r\app z$};

  \node[typebig] (betatype) at ($(alphatype) + (0.8*\defaulttypedistance, 0)$) {};
  \pic (betalabel2) [above=\labeldistance of betatype,xshift=-0.5cm] {labelrel={leftlabelstyle2}{$\inflerel{L_2}$}};
  \pic (betalabel1) [above=1.5*\typeequivlabeldistance of betalabel2equivclass] {labelrel={rightstyle}{$\inflerel{R_1}$}};

  \node[element] (l1x) at ($(x) + (0.8*\smallbigtypedistance,0)$) {};
  \node[above=\labeldistance of l1x] {$l_1\app x$};
  \node[element] (y) at ($(l1x) + (0, -\equivdistance)$) {};
  \node[left=\labeldistance of y] {$y$};
  \node[element] (y') at ($(y) + (0, -\equivdistance)$) {};
  \node[left=\labeldistance of y'] {$y'$};
  \node[element] (counit1r2z) at ($(y') + (0, -\equivdistance)$) {};
  \node[below=\labeldistance of counit1r2z] {$\counit_1\app (r_2\app z)$};

  \node[typesmall] (gammatype) at ($(betatype) + (\smallbigtypedistance, 0)$) {};
  \pic (alphalabel) [above=\labeldistance of gammatype,xshift=-0.6cm] {labelrel={rightlabelstyle2}{$\inflerel{R_2}$}};

  \node[element] (z) at ($(rz) + (2*\smallbigtypedistance,0)$) {};
  \node[right=\labeldistance of z] {$z$};

  \node[element] (r2z) at ($(z) - (0.8*\smallbigtypedistance,0)$) {};
  \node[below=\labeldistance of r2z] {$r_2\app z$};
  \node[element] (w') at ($(r2z) - (0, -\equivdistance)$) {};
  \node[right=\labeldistance of w'] {$w'$};
  \node[element] (w) at ($(w') - (0, -\equivdistance)$) {};
  \node[right=\labeldistance of w] {$w$};

  \draw[->,relstyle,rightstyle] (l1x) to (y);
  \draw[->,leftstyle2,relstyle2] (y) to (y');
  \draw[->,relstyle,rightstyle] (y') to (counit1r2z);
  \draw[->,leftstyle2,relstyle2] (w') to (r2z);
  \draw[->,relstyle,rightstyle] (w) to (w');
  \pic {leftedge={x}{l1x}{$l_1$}};
  \pic {leftedge={rz}{counit1r2z}{}};
  \pic {rightedgestyled={z}{r2z}{$r_2$}{rightstyle2}};

\end{tikzpicture}
}
  \caption{The initial setup of the proof.}
\end{subfigure}
\hfill
\begin{subfigure}[t]{0.49\textwidth}
  \centering
\resizebox{\textwidth}{!}{%
  \begin{tikzpicture} 
  \node[typesmall] (alphatype) {};
  \pic (alphalabel) [above=\labeldistance of alphatype,xshift=-0.6cm] {labelrel={leftstyle}{$\inflerel{L_1}$}};

  \node[element] (x) at ($(alphatype) + (0, 1.4)$) {};
  \node[above=\labeldistance of x] {$x$};
  \node[element] (rz) at ($(x) + (0, -3*\equivdistance)$) {};
  \node[below=\labeldistance of rz] {$r\app z$};

  \node[typebig] (betatype) at ($(alphatype) + (0.8*\defaulttypedistance, 0)$) {};
  \pic (betalabel2) [above=\labeldistance of betatype,xshift=-0.5cm] {labelrel={leftlabelstyle2}{$\inflerel{L_2}$}};
  \pic (betalabel1) [above=1.5*\typeequivlabeldistance of betalabel2equivclass] {labelrel={rightstyle}{$\inflerel{R_1}$}};

  \node[element] (l1x) at ($(x) + (0.8*\smallbigtypedistance,0)$) {};
  \node[above=\labeldistance of l1x] {$l_1\app x$};
  \node[element] (y) at ($(l1x) + (0, -\equivdistance)$) {};
  \node[left=\labeldistance of y] {$y$};
  \node[element] (y') at ($(y) + (0, -\equivdistance)$) {};
  \node[left=\labeldistance of y'] {$y'$};
  \node[element] (counit1r2z) at ($(y') + (0, -\equivdistance)$) {};
  \node[below=\labeldistance of counit1r2z] {$\counit_1\app (r_2\app z)$};

  \node[typesmall] (gammatype) at ($(betatype) + (\smallbigtypedistance, 0)$) {};
  \pic (alphalabel) [above=\labeldistance of gammatype,xshift=-0.6cm] {labelrel={rightlabelstyle2}{$\inflerel{R_2}$}};

  \node[element] (z) at ($(rz) + (2*\smallbigtypedistance,0)$) {};
  \node[right=\labeldistance of z] {$z$};

  \node[element] (r2z) at ($(z) - (0.8*\smallbigtypedistance,0)$) {};
  \node[below=\labeldistance of r2z] {$r_2\app z$};
  \node[element] (w') at ($(r2z) - (0, -\equivdistance)$) {};
  \node[right=\labeldistance of w'] {$w'$};
  \node[element] (w'') at ($(w') - (-\equivdistance, 0)$) {};
  \node[below=\labeldistance of w''] {$w''$};
  \node[element] (w) at ($(w') - (0, -\equivdistance)$) {};
  \node[right=\labeldistance of w] {$w$};

  \draw[->,relstyle,rightstyle] (l1x) to (y);
  \draw[->,leftstyle2,relstyle2] (y) to (y');
  \draw[->,relstyle,rightstyle] (y') to (counit1r2z);
  \draw[->,leftstyle2,relstyle2] (w') to (r2z);
  \draw[->,relstyle,rightstyle] (w) to (w');
  \draw[->,leftstyle2,relstyle2,proofedge] (w) to (w'');
  \draw[->,relstyle,rightstyle,proofedge] (w'') to (r2z);
  \draw[->,relstyle,rightstyle,proofedge] (counit1r2z) to (r2z);
  \pic {leftedge={x}{l1x}{$l_1$}};
  \pic {leftedge={rz}{counit1r2z}{}};
  \pic {rightedgestyled={z}{r2z}{$r_2$}{rightstyle2}};

\end{tikzpicture}
}
  \caption{Applying the compatibility condition to obtain $w''$.}
\end{subfigure}
\par\bigskip
\begin{subfigure}[t]{0.49\textwidth}
  \centering
\resizebox{\textwidth}{!}{%
  \begin{tikzpicture} 
  \node[typesmall] (alphatype) {};
  \pic (alphalabel) [above=\labeldistance of alphatype,xshift=-0.6cm] {labelrel={leftstyle}{$\inflerel{L_1}$}};

  \node[element] (x) at ($(alphatype) + (0, 1.4)$) {};
  \node[above=\labeldistance of x] {$x$};
  \node[element] (rz) at ($(x) + (0, -3*\equivdistance)$) {};
  \node[below=\labeldistance of rz] {$r\app z$};

  \node[typebig] (betatype) at ($(alphatype) + (0.8*\defaulttypedistance, 0)$) {};
  \pic (betalabel2) [above=\labeldistance of betatype,xshift=-0.5cm] {labelrel={leftlabelstyle2}{$\inflerel{L_2}$}};
  \pic (betalabel1) [above=1.5*\typeequivlabeldistance of betalabel2equivclass] {labelrel={rightstyle}{$\inflerel{R_1}$}};

  \node[element] (l1x) at ($(x) + (0.8*\smallbigtypedistance,0)$) {};
  \node[above=\labeldistance of l1x] {$l_1\app x$};
  \node[element] (y) at ($(l1x) + (0, -\equivdistance)$) {};
  \node[left=\labeldistance of y] {$y$};
  \node[element] (y') at ($(y) + (0, -\equivdistance)$) {};
  \node[left=\labeldistance of y'] {$y'$};
  \node[element] (counit1r2z) at ($(y') + (0, -\equivdistance)$) {};
  \node[below=\labeldistance of counit1r2z] {$\counit_1\app (r_2\app z)$};

  \node[typesmall] (gammatype) at ($(betatype) + (\smallbigtypedistance, 0)$) {};
  \pic (alphalabel) [above=\labeldistance of gammatype,xshift=-0.6cm] {labelrel={rightlabelstyle2}{$\inflerel{R_2}$}};

  \node[element] (z) at ($(rz) + (2*\smallbigtypedistance,0)$) {};
  \node[right=\labeldistance of z] {$z$};

  \node[element] (r2z) at ($(z) - (0.8*\smallbigtypedistance,0)$) {};
  \node[below=\labeldistance of r2z] {$r_2\app z$};
  \node[element] (w') at ($(r2z) - (0, -\equivdistance)$) {};
  \node[right=\labeldistance of w'] {$w'$};
  \node[element] (w'') at ($(w') - (-\equivdistance, 0)$) {};
  \node[below=\labeldistance of w''] {$w''$};
  \node[element] (w) at ($(w') - (0, -\equivdistance)$) {};
  \node[right=\labeldistance of w] {$w$};

  \node[element] (y'') at ($(y)!0.5!(r2z)$) {};
  \node[above=\labeldistance of y'',xshift=0.2cm] {$y''$};

  \draw[->,relstyle,rightstyle] (l1x) to (y);
  \draw[->,leftstyle2,relstyle2] (y) to (y');
  \draw[->,relstyle,rightstyle] (y') to (counit1r2z);
  \draw[->,leftstyle2,relstyle2] (w') to (r2z);
  \draw[->,relstyle,rightstyle] (w) to (w');
  \draw[->,leftstyle2,relstyle2] (w) to (w'');
  \draw[->,relstyle,rightstyle] (w'') to (r2z);
   \draw[->,relstyle,rightstyle] (counit1r2z) to (r2z);
  \draw[->,relstyle,rightstyle,proofedge] (y) to (y'');
  \draw[->,leftstyle2,relstyle2,proofedge] (y'') to (r2z);
  \pic {leftedge={x}{l1x}{$l_1$}};
  \pic {leftedge={rz}{counit1r2z}{}};
  \pic {rightedgestyled={z}{r2z}{$r_2$}{rightstyle2}};

\end{tikzpicture}
}
  \caption{Applying the compatibility condition to obtain $y''$.}
\end{subfigure}
\hfill
\begin{subfigure}[t]{0.49\textwidth}
  \centering
\resizebox{\textwidth}{!}{%
  \begin{tikzpicture} 
  \node[typesmall] (alphatype) {};
  \pic (alphalabel) [above=\labeldistance of alphatype,xshift=-0.6cm] {labelrel={leftstyle}{$\inflerel{L_1}$}};

  \node[element] (x) at ($(alphatype) + (0, 1.4)$) {};
  \node[above=\labeldistance of x] {$x$};
  \node[element] (rz) at ($(x) + (0, -3*\equivdistance)$) {};
  \node[below=\labeldistance of rz] {$r\app z$};

  \node[typebig] (betatype) at ($(alphatype) + (0.8*\defaulttypedistance, 0)$) {};
  \pic (betalabel2) [above=\labeldistance of betatype,xshift=-0.5cm] {labelrel={leftlabelstyle2}{$\inflerel{L_2}$}};
  \pic (betalabel1) [above=1.5*\typeequivlabeldistance of betalabel2equivclass] {labelrel={rightstyle}{$\inflerel{R_1}$}};

  \node[element] (l1x) at ($(x) + (0.8*\smallbigtypedistance,0)$) {};
  \node[above=\labeldistance of l1x] {$l_1\app x$};
  \node[element] (y) at ($(l1x) + (0, -\equivdistance)$) {};
  \node[left=\labeldistance of y] {$y$};
  \node[element] (y') at ($(y) + (0, -\equivdistance)$) {};
  \node[left=\labeldistance of y'] {$y'$};
  \node[element] (counit1r2z) at ($(y') + (0, -\equivdistance)$) {};
  \node[below=\labeldistance of counit1r2z] {$\counit_1\app (r_2\app z)$};

  \node[typesmall] (gammatype) at ($(betatype) + (\smallbigtypedistance, 0)$) {};
  \pic (alphalabel) [above=\labeldistance of gammatype,xshift=-0.6cm] {labelrel={rightlabelstyle2}{$\inflerel{R_2}$}};

  \node[element] (z) at ($(rz) + (2*\smallbigtypedistance,0)$) {};
  \node[right=\labeldistance of z] {$z$};

  \node[element] (r2z) at ($(z) - (0.8*\smallbigtypedistance,0)$) {};
  \node[below=\labeldistance of r2z] {$r_2\app z$};
  \node[element] (w') at ($(r2z) - (0, -\equivdistance)$) {};
  \node[right=\labeldistance of w'] {$w'$};
  \node[element] (w'') at ($(w') - (-\equivdistance, 0)$) {};
  \node[below=\labeldistance of w''] {$w''$};
  \node[element] (w) at ($(w') - (0, -\equivdistance)$) {};
  \node[right=\labeldistance of w] {$w$};

  \node[element] (y'') at ($(y)!0.5!(r2z)$) {};
  \node[above=\labeldistance of y'',xshift=0.2cm] {$y''$};
  \node[element] (u) at ($(y) - (\equivdistance,0)$) {};
  \node[below=\labeldistance of u] {$u$};

  \node[element] (lx) at ($(x) + (2*\smallbigtypedistance,0)$) {};
  \node[right=\labeldistance of lx] {$lx$};

  \draw[->,relstyle,rightstyle] (l1x) to (y);
  \draw[->,leftstyle2,relstyle2] (y) to (y');
  \draw[->,relstyle,rightstyle] (y') to (counit1r2z);
  \draw[->,leftstyle2,relstyle2] (w') to (r2z);
  \draw[->,relstyle,rightstyle] (w) to (w');
  \draw[->,leftstyle2,relstyle2] (w) to (w'');
  \draw[->,relstyle,rightstyle] (w'') to (r2z);
   \draw[->,relstyle,rightstyle] (counit1r2z) to (r2z);
  \draw[->,relstyle,rightstyle] (y) to (y'');
  \draw[->,leftstyle2,relstyle2] (y'') to (r2z);
  \pic {leftedge={x}{l1x}{$l_1$}};
  \pic {leftedge={rz}{counit1r2z}{}};
  \pic {rightedgestyled={z}{r2z}{$r_2$}{rightstyle2}};
  \draw[->,relstyle2,leftstyle2,proofedge] (l1x) to (u);
  \draw[->,relstyle,rightstyle,proofedge] (u) to (y');
  \draw[->,loopstyle,leftstyle2,relstyle] (l1x) to [out=330,in=300] (l1x);
  \draw[->] (l1x) edge[leftstyle2,edgestyle,"$l_2$",pos=0.65] (lx);
  \draw[->,loop below, loopstyle,rightstyle2,relstyle2] (lx) to (lx);

\end{tikzpicture}
}
\caption{Applying the compatibility condition to show $\indom\app\inflerel{L_2}\app (l_1\app x)$. Then apply reflexivity of $\inflerel{L_2}$ and monotonicity of $l_2$ to finish.}
\end{subfigure}
\caption{Proof of $\halfgalpl{\inflerel{L}}{\inflerel{R}}{l}{r}$ as explained in \cref{thm:comp_galequiv_better}.
Types are drawn solid, black,
transport functions dashed,
relations dotted and dashed-dotted.
}\label{fig:comp_galequiv}
\Description[Proof of $\halfgalpl{\inflerel{L}}{\inflerel{R}}{l}{r}$ as explained in \cref{thm:comp_galequiv}.]{$\halfgalpl{\inflerel{L}}{\inflerel{R}}{l}{r}$ as explained in \cref{thm:comp_galequiv}. The chase of $\inflerel{L}$ can be transformed to a chase for $\infelrel{R}$ using the compatibility condition on the middle relations.}
\end{figure}
Assume that
\begin{assumes}[label=\alph*,ref=\alph*](2)
\item\label{asm:comp_incodom}$\incodom\app\inflerel{R}\app z$,
\item\label{asm:comp_le}$x \lerel{L} r\app z$.
\end{assumes}
We have to show that $l\app x \lerel{R} z$, which unfolds to
$\parenths[\big]{\infgalrel{R_2} \relcomp \inflerel{R_1} \relcomp \infgalrel{L_2}}\app (l_2\app(l_1\app x))\app z$.
From \cref{asm:comp_le},
we obtain $y,y'$ such that
\begin{displaymath}
l_1\app x\lerel{R_1}y\lerel{L_2}y'\lerel{R_1} l_1\app (r\app z)=\epsilon_1\app (r_2\app z),
\end{displaymath}
where \makebox{$\counit_1\define \counitconst\app l_1\app r_1$}.
We wish to obtain $\counit_1\app (r_2\app z)\lerel{R_1} r_2\app z$;
this only holds if $\incodom\app \inflerel{R_1} (r_2\app z)$, however.
For this purpose, take \cref{asm:comp_galequiv,asm:comp_incodom}.
We obtain $w,w'$ such that $w\lerel{R_1}w'\lerel{L_2} r_2\app z$.
Thus, by \cref{asm:comp_middlecompatcond}, there is $w''$ such that
$w\lerel{L_2}w''\lerel{R_1} r_2\app z$.
Hence, $\incodom\app \inflerel{R_1} (r_2\app z)$.

Then by transitivity, we get $y\lerel{L_2}y'\lerel{R_1} r_2\app z$.
Thus, by \cref{asm:comp_middlecompatcond}, there is $y''$ such that
$y\lerel{R_1}y''\lerel{L_2} r_2\app z$.
From $y''\lerel{L_2} r_2\app z$ and \cref{asm:comp_incodom},
we get $y''\galrel{L_2} z$.
From $l_1\app x\lerel{R_1}\app y\lerel{R_1}\app y''$ and transitivity,
we get $l_1\app x\lerel{R_1} y''$.
It remains to show that $l\app x\galrel{R_2}l_1\app x$,
that is $l\app x\lerel{R_2}l\app x$ and
$\incodom\app\inflerel{L_2}\app (l_1\app x)$.

From $l_1\app x\lerel{R_1}y\lerel{L_2}y'$ and
\cref{asm:comp_middlecompatcond},
we obtain $u$ such that
$l_1\app x\lerel{L_2}u\lerel{R_1}y'$.
Thus, $\indom\app\inflerel{L_2}\app (l_1\app x)$.
Then by reflexivity (\cref{asm:comp_preorderleft}), $l_1\app x\lerel{L_2}l_1\app x$.
Finally, $l\app x\lerel{R_2}l\app x$ by monotonicity of $l_2$ (\cref{asm:comp_galequiv}).
\end{proof}

We can also prove a generalisation of
\cref{thm:comp_galreleq}:
\begin{theorem}\label{thm:comp_galreleq_gal_con}
Assume
\begin{assumes}(2)
\item $\parenths[\big]{\mono{\inflerel{R_1}}{\inflerel{L_1}}}\app r_1$,
\item$\galp[\big]{\inflerel{L_1}}{\inflerel{R_1}}{l_1}{r_1}$,
\item$\halfgalpl[\big]{\inflerel{R_1}}{\inflerel{L_1}}{r_1}{l_1}$,
\item $\preorderon\app\infieldapp{\inflerel{R_1}}$,
\item $\parenths[\big]{\mono{\inflerel{L_2}}{\inflerel{R_2}}}\app l_2$,
\item$\halfgalpl[\big]{\inflerel{R_2}}{\inflerel{L_2}}{r_2}{l_2}$,
\item $\reflon\app(\indom\app \inflerel{L_2})\app \inflerel{L_2}$,
\item $\parenths[\big]{\inflerel{R_1} \relcomp \inflerel{L_2}} = \parenths[\big]{\inflerel{L_2} \relcomp \inflerel{R_1}}$.
\end{assumes}
Then
$\infgalrel{L}=\parenths[\big]{\infgalrel{L_1}\relcomp\infgalrel{L_2})}$.
\end{theorem}

\subsubsection{Comparison To Isabelle's Lifting Package}\label{sec:appendixcomparisoncomplifting}
As mentioned, our definitions are
inspired by~\cite{isabellelifting}:
Let $(T_1,l_1,r_1)$ and $(T_2,l_2,r_2)$
be two partial quotient types
with induced left relations $\inflerel{L_1}$ and $\inflerel{L_2}$.
Huffman and and Kun{\v{c}}ar then construct the composition
$(T_1\relcomp T_2, l_2\circ l_1, r_1\circ r_2)$.
Moreover, they prove that the induced left relation $\inflerel{L}$
of this composed partial quotient type satisfies
$\inflerel{L} = T_1\relcomp \inflerel{L_2}\relcomp \inv{T_1}$.
This insight sparked the idea of our definitions.

Indeed, we can show that our definitions
faithfully generalise their work.
Just as \cref{lem:galrelpartquoteq}
shows that
$T_1 = \galrelconst\app \inflerel{L_1}\app (=)\app r_1$,
we can show that
$\inv{T_1} = \galrelconst\app (=)\app \inflerel{L_1}\app l_1$.
It then follows that
\begin{displaymath}
\parenths[\big]{T_1\relcomp \inflerel{L_2}\relcomp \inv{T_1}} =
\parenths[\big]{\galrelconst\app \inflerel{L_1}\app (=)\app r_1\relcomp \inflerel{L_2}\relcomp \galrelconst\app (=)\app \inflerel{L_1}\app l_1}.
\end{displaymath}
Moreover, it is easy to show that the compatibility condition
is vacuously true for partial quotient types.

\end{document}